\documentclass{beatcs}
\usepackage{amsfonts,amssymb,amsmath,amsthm,graphicx}

\newcommand{\card}[1]{\left\lvert#1\right\rvert}
\newcommand{\No}{\textbf{No}}
\newcommand{\reals}{\mathbb{R}}
\newcommand{\lapprox}{\le}
\newcommand{\gapprox}{\ge}

\newcommand{\tuple}[1]{{\langle{#1}\rangle}}
\newcommand{\disjointunion}{\mathop{\dot{\cup}}}
\newcommand{\union}{\mathop{\cup}}

\newcommand{\identity}{\textup{id}}
\newcommand{\map}[3]{{{#1}:{#2}\rightarrow{#3}}}
\newcommand{\p}{\varphi}
\newcommand{\nums}{\mathbb{N}}
\newcommand{\gset}{\textit{g-set}}
\newcommand{\gnum}{g}
\newcommand{\eqdf}{\mathrel{:=}}
\newcommand{\PSPACE}{\textbf{\textup{PSPACE}}}
\newcommand{\AC}{\textbf{\textup{AC}}}
\newcommand{\NC}{\textbf{\textup{NC}}}

\newcommand{\LOGSPACE}{\textbf{\textup{L}}}
\newcommand{\NL}{\textbf{\textup{NL}}}
\newcommand{\PTIME}{\textbf{\textup{P}}}
\newcommand{\NP}{\textbf{\textup{NP}}}

\newcommand{\PP}{\textbf{\textup{PP}}}

\newcommand{\CeqP}{\mathord{\textbf{C}_=\PTIME}}
\newcommand{\FO}{\textbf{\textup{FO}}}
\newcommand{\PG}{\textup{PG}}
\newcommand{\EX}{\textup{Explicit}}
\newcommand{\SU}{\textup{Succinct}}
\newcommand{\PO}{\textup{PO}}
\newcommand{\HD}{\textup{HD}}
\newcommand{\AR}{\textup{AR}}
\newcommand{\co}[1]{\mathord{\textbf{\textup{co}}{#1}}}
\newcommand{\threshold}[2]{{\textup{Threshold}({#1},{#2})}}
\newcommand{\myand}{\mathrel{\wedge}}

\newcommand{\two}{{\{0,1\}}}
\newcommand{\cP}{\mathcal{P}}
\newcommand{\cN}{\mathcal{N}}
\newcommand{\cR}{\mathcal{R}}
\newcommand{\cL}{\mathcal{L}}
\newcommand{\bbG}{\textbf{PG}}

\newcommand{\Hex}{\mbox{\sc Hex}}
\newcommand{\Nim}{\mbox{\sc Nim}}
\newcommand{\Chomp}{\mbox{\sc Chomp}}
\newcommand{\Divisors}{\mbox{\sc Divisors}}
\newcommand{\Hackendot}{\mbox{\sc Hackendot}}

\newcommand{\Geo}{\mbox{\sc Geography}}
\newcommand{\NK}{\mbox{\sc Node Kayles}}
\newcommand{\Col}{\mbox{\sc Col}}
\newcommand{\Nfree}{\mbox{\sc N-Free}}
\newcommand{\GenCol}{\mbox{\sc GenCol}}

\theoremstyle{definition}
\newtheorem{definition}{Definition}[section]
\newtheorem{exercise}[definition]{Exercise}

\theoremstyle{plain}
\newtheorem{theorem}[definition]{Theorem}
\newtheorem{lemma}[definition]{Lemma}
\newtheorem{fact}[definition]{Fact}
\newtheorem{claim}[definition]{Claim}
\newtheorem{corollary}[definition]{Corollary}

\newtheorem{proposition}[definition]{Proposition}

\DeclareMathOperator{\mex}{mex}

\title{Combinatorial Game Complexity: An Introduction with Poset Games}
\author{
  Stephen A. Fenner\thanks{University of South Carolina, Computer Science and Engineering Department.  Technical report number CSE-TR-2015-001.} 
\and 
  John Rogers\thanks{DePaul University, School of Computing}
}

\begin{document}

\maketitle

%

\begin{abstract}
Poset games have been the object of mathematical study for over a century, but little has been written on the computational complexity of determining important properties of these games.  In this introduction we develop the fundamentals of combinatorial game theory and focus for the most part on poset games, of which \Nim\ is perhaps the best-known example.  We present the complexity results known to date, some discovered very recently.
\end{abstract}


\section{Introduction}
\label{sec:intro}

Combinatorial games have long been studied (see \cite{Conway:ONAG,BCG:WW}, for example) but the record of results on the complexity of questions arising from these games is rather spotty.  Our goal in this introduction is to present several results---some old, some new---addressing the complexity of the fundamental problem given an instance of a combinatorial game:
\begin{quote}
Determine which player has a winning strategy.
\end{quote}
A secondary, related problem is
\begin{quote}
Find a winning strategy for one or the other player, or just find a winning first move, if there is one.
\end{quote}
The former is a decision problem and the latter a search problem.  In some cases, the search problem clearly reduces to the decision problem, i.e., having a solution for the decision problem provides a solution to the search problem.  In other cases this is not at all clear, and it may depend on the class of games you are allowed to query.

We give formal definitions below, but to give an idea of the subject matter, we will discuss here the large class of games known as the \emph{poset games}.  One of the best known of these is \Nim, an ancient game, but given its name by Charles Bouton in 1901 \cite{Bouton:Nim}.  There are many others, among them, Hackendot, Divisors, and Chomp \cite{Conway:ONAG}.  Poset games not only provide good examples to illustrate general combinatorial game concepts, but they also are the subject of a flurry of recent results in game complexity, which is the primary focus of this article.

The rest of this section gives some basic techniques for analyzing poset games.  Section~\ref{sec:definitions} lays out the foundations of the general theory of combinatorial games, including numeric and impartial games, using poset games as examples.  The rest of the paper is devoted to computational complexity.  Section~\ref{sec:ub} gives an upper bound on the complexity of so-called ``N-free'' games, showing that they are solvable in polynomial time.  Section~\ref{sec:lb} gives lower bounds on the complexity of some games, showing they are hard for various complexity classes.  The section culminates in two recent $\PSPACE$-completeness results---one for impartial poset games, and the other for ``black-white'' poset games.  Section~\ref{sec:open} discusses some open problems.

\subsection{Poset games}

\begin{definition}\rm
A \emph{partial order} on a set $P$ (hereafter called a \emph{poset}) is a binary relation $\le$ on $P$ that is reflexive, transitive, and antisymmetric (i.e., $x\le y$ and $y\le x$ imply $x=y$).  For any $x\in P$, define $P_x \eqdf \{ y\in P \mid x\not\le y \}$.

We identify a finite poset $P$ with the corresponding \emph{poset game}: Starting with $P$, two players (Alice and Bob, say) alternate moves, Alice moving first, where a move consists of choosing any point $x$ in the remaining poset and removing all $y$ such that $x\le y$, leaving $P_x$ remaining.  Such a move we call \emph{playing $x$}.  The first player unable to move (because the poset is empty) loses.\footnote{Games can be played on some infinite posets as well, provided every possible sequence of moves is finite.  This is true if and only if the poset is a well-quasi-order (see, e.g., Kruskal \cite{Kruskal:wqo}).}
\end{definition}

Poset games are \emph{impartial}, which means that, at any point in the play, the set of legal moves is the same for either player.  There is a rich theory of impartial games, and we cover it in Section~\ref{sec:impartial}.

In an impartial game, the only meaningful distinction between players is who plays first (and we have named her Alice).  Since every play of a poset game has only finitely many moves, one of the two players (but clearly not both!) must have a winning strategy.  We say that a poset $P$ is an \emph{$\exists$-game} (or \emph{winning position}) if the first player has a winning strategy, and $P$ is a \emph{$\forall$-game} (or \emph{losing position}) if the second player has a winning strategy.  In the combinatorial game theory literature, these are often called $\mathcal N$-games (``Next player win'') and $\mathcal P$-games (``Previous player win''), respectively.  We get the following concise inductive definition for any poset $P$:
\begin{verse}
$P$ is an $\exists$-game iff there exists $x\in P$ such that $P_x$ is a $\forall$-game.\\
$P$ is a $\forall$-game iff $P$ is not an $\exists$-game (iff, for all $x\in P$, \ $P_x$ is an $\exists$-game).
\end{verse}
We call the distinction of a game being a $\forall$-game versus an $\exists$-game the \emph{outcome} of the game.

There are at least two natural ways of combining two posets to produce a third.

\begin{definition}\label{def:series-parallel}\rm
For posets $P = \tuple{P,\le_P}$ and $Q=\tuple{Q,\le_Q}$,
\begin{itemize}
\item
define $P+Q$ (the \emph{parallel union of $P$ and $Q$}) to be the disjoint union of $P$ and $Q$, where all points in $P$ are incomparable with all points in $Q$:
\[ P + Q \eqdf \tuple{P \disjointunion Q, \le}\;, \]
where $\le \eqdf \le_P \disjointunion \le_Q$.
\item
Define $P/Q$ (or $\frac{P}{Q}$---the \emph{series union of $P$ over $Q$}) to be the disjoint union of $P$ and $Q$ where all points in $P$ lie above (i.e., are $\ge$ to) all points in $Q$:
\[ \frac{P}{Q} \eqdf \tuple{P \disjointunion Q, \le}\;, \]
where $\mathord{\le} \eqdf \mathord{\le_P} \disjointunion \mathord{\le_Q} \disjointunion (Q\times P)$.
\end{itemize}
\end{definition}

Note that $+$ is commutative and associative, and that $/$ is associative but not commutative.  Using these two operations, let's build some simple posets.  Let $C_1$ be the one-element poset.  For any $n\in\nums$, let
\begin{enumerate}
\item
$C_n \eqdf \underbrace{C_1/C_1/\ldots/C_1}_n$ is the chain of $n$ points (totally ordered).  This is also called a \emph{NIM stack}.
\item
$A_n \eqdf \underbrace{C_1 + C_1 + \cdots + C_1}_n$ is the antichain of $n$ pairwise incomparable points.
\item
$V_n \eqdf A_n/C_1$ is the $n$-antichain with a common lower bound.
\item
$\Lambda_n \eqdf C_1/A_n$ is the $n$-antichain with a common upper bound.
\item
$\Diamond_n \eqdf C_1/A_n/C_1$ is the $n$-antichain with common upper and lower bounds.
\end{enumerate}
Some examples are shown in Figure~\ref{fig:sample-games}.
\begin{figure}
\begin{center}
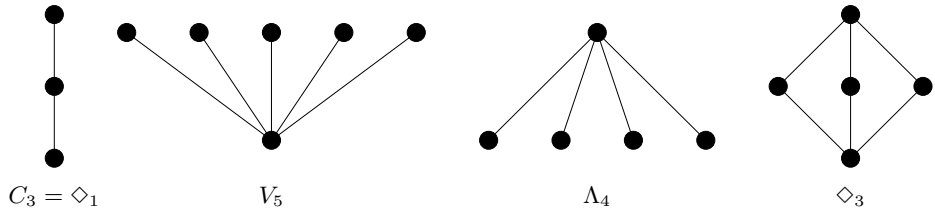
\caption{Some simple posets constructed from individual points via parallel and series union.}\label{fig:sample-games}
\end{center}
\end{figure}

\begin{exercise}
Find a simple way, given $m$ and $n$, to determine whether $A_m/A_n$ is an $\exists$-game or a $\forall$-game.
\end{exercise}

\begin{exercise}\label{ex:OR-gate}
Show that $P/Q$ is an $\exists$-game if and only if either $P$ is an $\exists$-game or $Q$ is an $\exists$-game.
\end{exercise}

\subsubsection{More examples}
\label{sec:poset-game-examples}

The best-known poset game is \Nim, an ancient game first formally described and ``solved'' by C. L. Bouton in 1902 \cite{Bouton:Nim}.  Here, the poset is a union of disjoint chains, i.e., of the form $C_{n_1} + C_{n_2} + \cdots + C_{n_k}$ for some positive integers $n_1,\ldots,n_k$.  A move then consists of choosing a point in one of the chains and remove that point and everything above it.

Other families of poset games include
\begin{description}
\item[\Chomp,] introduced in 1974 by D. Gale \cite{Gale:Chomp}, which, in its finite form, is represented by a rectangular arrangement of squares with the leftmost square in the bottom row removed.  This is a poset with two minimal elements (first square on the second row, second square on bottom row).  Every element in a row is greater than all of the elements to the left and below so playing an element removes it and all elements to the right and above.
\item[\Hackendot,] attributed to von Newmann, where the poset is a forest of upside-down trees (roots at the top).  \Hackendot\ was solved in 1980 by \'Ulehla \cite{Ulehla:Hackendot}.
\item[\Divisors,] introduced by F. Schuh \cite{Schuh:divisors}, the poset is the set of all positive divisors (except $1$) of a fixed integer $n$, partially ordered by divisibility.  \Divisors\ is a multidimensional generalization of \Chomp.  \Chomp\ occurs as the special case where $n = p^mq^n$ for distinct primes $p,q$.
\end{description}

\subsection{Dual symmetry}

Some poset games can be determined (as $\exists$-games or $\forall$-games just by inspection).  For example, suppose a poset $P$ has some kind of dual symmetry, that is, there is an order-preserving map $\map{\p}{P}{P}$ such that $\p\circ\p = \identity$.

\begin{fact}\label{fact:dual-symmetry}
Let $P$ be a poset and let $\map{\p}{P}{P}$ be such that $\p\circ\p = \identity_P$ and $x\le y \iff \p(x) \le \p(y)$ for all $x,y\in P$.
\begin{itemize}
\item
If $\p$ has no fixed points, then $P$ is a $\forall$-game.
\item
If $\p$ has a minimum fixed point (minimum among the set of fixed points), then $P$ is an $\exists$-game.
\end{itemize}
\end{fact}

\begin{proof}
If $\p$ has no fixed points, then Bob can answer any $x$ played by Alice by playing $\p(x)$.  If $\p$ has a least fixed point $z$, then Alice plays $z$ on her first move, leaving $P_z$, which is symmetric with no fixed points, and thus a $\forall$-game.
\end{proof}

For example, the poset below is symmetric with a unique fixed point $x$, which Alice can win by playing on her first move:
\begin{center}
\begin{picture}(0,0)%
\includegraphics{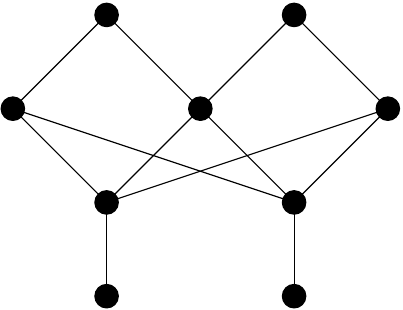}%
\end{picture}%
\setlength{\unitlength}{2960sp}%
\begingroup\makeatletter\ifx\SetFigFont\undefined%
\gdef\SetFigFont#1#2#3#4#5{%
  \reset@font\fontsize{#1}{#2pt}%
  \fontfamily{#3}\fontseries{#4}\fontshape{#5}%
  \selectfont}%
\fi\endgroup%
\begin{picture}(2566,1966)(2318,-3144)
\put(3751,-1936){\makebox(0,0)[lb]{\smash{{\SetFigFont{9}{10.8}{\rmdefault}{\mddefault}{\updefault}$x$}}}}
\end{picture}%

\end{center}
After we introduce game equivalence, we can give a partial generalization of Fact~\ref{fact:dual-symmetry} (Lemma~\ref{lem:dual-symmetry} below) that has been useful in determining the outcomes of several games.

\subsection{Strategy stealing}

Another class of posets that are easy to determine by inspection are those with an \emph{articulation point}, i.e., a point that is comparable with every other point in the poset.  For example, minimum and maximum points of $P$ are articulation points.

\begin{fact}
If a poset $P$ contains an articulation point, then $P$ is an $\exists$-game.
\end{fact}

\begin{proof}
Let $x$ be some articulation point of $P$.  If $x$ is a winning first move for Alice, then we are done.  If $x$ is a losing first move for Alice, then there must be some winning response $y$ for Bob if Alice first plays $x$.  But if Alice plays $x$, then all points $\ge x$ are now gone, and so we have $y<x$.  This means that the game after Bob moves is $P_y$, which is a $\forall$-game by assumption.  But then, Alice could have played $y$ instead on her first move, leaving the $\forall$-game $P_y$ for Bob, and thus winning.
\end{proof}

We call this ``strategy stealing'' because Alice steals Bob's winning strategy.  The interesting thing about this proof is how nonconstructive it is.  It shows that Alice has a winning first move, but gives virtually no information about what that first move could be.  All we know is that the winning first play must be $\le x$.  If $x$ is a maximum point of $P$, then the proof gives no information whatsoever about Alice's winning first move.  Several poset games, including \Chomp, have initial posets with maximum points, so we know that they are $\exists$-games.  But determining a winning first move for Alice in \Chomp\ appears quite difficult, and no fast algorithm is known.  This suggests that, in the case of \Chomp\ at least, the search problem (finding a winning first move) is apparently difficult, whereas the decision problem ($\exists$-game or $\forall$-game?) is trivial.  The search versus decision issue is discussed further in Section~\ref{sec:representations}, below.

\begin{exercise}
Show that the winning first moves in any poset form an antichain.
\end{exercise}

\subsection{Black-white poset games}

Many interesting games are not impartial because the legal moves differ for the players.  In chess, for example, one player can only move white pieces and the other only black pieces.  We will informally call a game ``black-white'' when each player is assigned a color (black or white) and can only make moves corresponding to their color.\footnote{A different, popular color combination is red-blue.  We use black-white so that illustrations are faithfully rendered on a black-and-white printer.}  Many impartial games have natural black-white versions.  Here, then, is a black-white version of a poset game:

\begin{definition}\label{def:bw-game}
A \emph{black-white poset game} consists of finite poset $P$, each of whose points are colored either black or white.  The same rules apply to black-white poset games as to (impartial) poset games, except that one player (Black) can only play black points and the other player (White) can only play white points.  (All points above a played point are still removed, regardless of color.)
\end{definition}

One could generalize this definition by allowing a third color, grey, say, where grey points can be played by either player.  We will not pursue this idea further.  Other ``colored'' games include red-blue Hackenbush and red-green-blue Hackenbush \cite{BCG:WW}.

Combinatorial games that are not impartial are known as \emph{partisan}.  In partisan games, we must make a distinction between the two players beyond who moves first.  Generically, these players are called Left and Right.  There is a surprisingly robust general theory of combinatorial games, both impartial and partisan, developed in $\cite{BCG:WW,Conway:ONAG}$, and we give the basics of this theory in the next section.



%
%


\section{Combinatorial game theory basics}
\label{sec:definitions}

In this section we give some relevant definitions and a few facts from the general theory of combinatorial games.  We give enough of the theory to understand later results.  Thorough treatments of this material, with lots of examples, can be found in \cite{BCG:WW,Conway:ONAG} as well as other sources, e.g., \cite{FHS:combinatorial-games} and the recent book by Siegel \cite{Siegel:combinatorial-games}.  Our terminology and notation vary a little bit from \cite{BCG:WW,Conway:ONAG}, but the concepts are the same.  When we say, ``game,'' we always mean what is commonly referred to as a \emph{combinatorial game}, i.e., a game between two players, say, Left and Right, alternating moves with perfect information, where the first player unable to move loses (and the other wins).  In their fullest generality, these games can be defined abstractly by what options each player has to move, given any position in the game.

\subsection{Notation}

We let $\nums$ denote the set $\{0,1,2,\ldots,\}$ of natural numbers.  We let $\card X$ denote the cardinality of a finite set $X$.  We use the relation ``$\eqdf$'' to mean ``equals by definition.''  We extend the definition of an operator on games to an operator on sets of games in the customary way; for example, if $*$ is a binary operation on games, and $G$ and $H$ are sets of games, then $G*H \eqdf \{g*h \mid g\in G \myand h\in H\}$, and if $g$ is a game, then $g*H \eqdf \{g\}*H$, and so on.

\subsection{Basic definitions}

\begin{definition}
A \emph{game} is an ordered pair $G = (G^L,G^R)$, where $G^L$ and $G^R$ are sets of games.  The elements of $G^L$ (respectively, $G^R$) are the \emph{left options} (respectively, \emph{right options}) of $G$.  An \emph{option} of $G$ is either a left option or a right option of $G$.
\end{definition}

It is customary to write $\{G^L|G^R\}$ or $\{\ell_1,\ell_2,\ldots|r_1,r_2,\ldots\}$ rather than $(G^L,G^R)$, where $G^L = \{\ell_1,\ell_2,\ldots\}$ and $G^R = \{r_1,r_2,\ldots\}$.  We will do the same.

For this and the following inductive definitions to make sense, we tacitly assume that the ``option of'' relation is well-founded, i.e., there is no infinite sequence of games $g_1,g_2,\ldots$ where $g_{i+1}$ is an option of $g_i$ for all $i$.\footnote{This follows from the Foundation Axiom of set theory, provided ordered pairs are implemented in some standard way, e.g., $(x,y) \eqdf \{\{x\},\{x,y\}\}$ for all sets $x$ and $y$.}  A \emph{position} of a game $G$ is any game reachable by making a finite series of moves starting with $G$ (the moves need not alternate left-right).  Formally,

\begin{definition}
A \emph{position} of a game $G$ is either $G$ itself or a position of some option of $G$.  We say that $G$ is \emph{finite} iff $G$ has a finite number of positions.\footnote{Finite games are sometimes called \emph{short games}; see \cite{Siegel:combinatorial-games}.}
\end{definition}

Starting with a game $G$, we imagine two players, Left and Right, alternating moves as follows: the initial position is $G$; given the current position $P$ of $G$ (also a game), the player whose turn it is chooses one of her or his options of $P$ (left options for Left; right options for Right), and this option becomes the new game position.  The first player faced with an empty set of options loses.  The sequence of positions obtained this way is a \emph{play} of the game $G$.  Our well-foundedness assumption implies that every play is finite, and so there must be a winning strategy for one or the other player.  We classify games by who wins (which may depend on who moves first) when the players play optimally.  This is our broadest and most basic classification.  Before giving it, we first introduce the ``mirror image'' of a game $G$: define $-G$ to be the game where all left options and right options are swapped at every position, as if the players switched places.  Formally,

\begin{definition}
For any game $G$, define $-G \eqdf \{{-G^R}|{-G^L}\}$.
\end{definition}

It is a good warm-up exercise to prove---inductively, of course---that $-(-G) = G$ for every game $G$.  For impartial games, e.g., poset games, the ``$-$'' operator has no effect; for black-white poset games, this is tantamount to swapping the color of each point in the poset.

We can consider the following definition to be the most fundamental property of a game:

\begin{definition}\label{def:ge-zero}
Let $G$ be a game.  We say that $G\gapprox 0$ (or $0\lapprox G$) iff there is no right option $g^R$ of $G$ such that $-g^R \gapprox 0$.  We will say $G\lapprox 0$ to mean that $-G\gapprox 0$.
\end{definition}

So $G\gapprox 0$ if and only if no right option $g^R$ of $G$ satisfies $g^R \lapprox 0$.  Symmetrically, $G\lapprox 0$ if and only if no left option $g^L$ of $G$ satisfies $g^L \gapprox 0$.  In terms of strategies, $G\gapprox 0$ means that $G$ is a \emph{first-move loss for Right} or a \emph{second-move win for Left}.  If Right has to move first in $G$, then Left can win.  Symmetrically, $G\lapprox 0$ means that $G$ is a \emph{first-move loss for Left} or a \emph{second-move win for Right}.

The $\lapprox$ notation suggests that a partial order (or at least, a preorder) on games is lurking somewhere.  This is true, and we develop it below.


Definition~\ref{def:ge-zero} allows us to partition all games into four broad categories.

\begin{definition}\label{def:outcome}
Let $G$ be a game.
\begin{itemize}
\item
$G$ is a \emph{zero game} (or a \emph{first-move loss}, or $\cP$-game) iff $G\lapprox 0$ and $G\gapprox 0$.
\item
$G$ is \emph{positive} (or a \emph{win for Left}, or $\cL$-game) iff $G\gapprox 0$ and $G\not\lapprox 0$.
\item
$G$ is \emph{negative} (or a \emph{win for Right}, or $\cR$-game) iff $G\lapprox 0$ and $G\not\gapprox 0$.
\item
$G$ is \emph{fuzzy} (or a \emph{first-move win}, or $\cN$-game) iff $G\not\lapprox 0$ and $G\not\gapprox 0$.
\end{itemize}
These four categories, $\cP$ (for previous player win), $\cL$ (for Left win), $\cR$ (for Right win), and $\cN$ (for next player win), partition the class of all games.  The unique category to which $G$ belongs is called the \emph{outcome} of $G$, written $o(G)$.
\end{definition}

For example, the simplest game is the \emph{endgame} $0 \eqdf \{|\}$ with no options, which is a zero game ($o(0) = \cP$).  The game $1 \eqdf \{0|\}$ is positive ($o(1) = \cL$), and the game $-1 \eqdf \{|0\}$ is negative $o(-1) = \cR$, while the game $* \eqdf \{0|0\}$ is fuzzy ($o(*) = \cN$).

\subsection{Game arithmetic, equivalence, and ordering}

Games can be added, and this is a fundamental construction on games.  The sum $G+H$ of two games $G$ and $H$ is the game where, on each move, a player may decide in which of the two games to play.  Formally:

\begin{definition}
Let $G$ and $H$ be games.  We define
\[ G+H \eqdf \{(G^L+H)\cup (G+H^L)\mid (G^R+H)\cup (G+H^R)\}\;. \]
\end{definition}

In Section~\ref{sec:intro} we used the $+$ operator for the parallel union of posets.  Observe that this corresponds exactly to the $+$ operator on the corresponding games, i.e., the game corresponding to the parallel union of posets $P$ and $Q$ is the game-theoretic $+$ applied to the corresponding poset \emph{games} $P$ and $Q$.

We write $G-H$ as shorthand for $G+(-H)$.  One can easily show by induction that $+$ is commutative and associative when applied to games, and the endgame $0$ is the identity under $+$.  This makes the class of all games into a commutative monoid (albeit a proper class).  One can also show for all games $G$ and $H$ that $-(G+H) = -G-H$.  Furthermore, if $G\ge 0$ and $H\ge 0$, then $G+H\ge 0$.  It is \emph{not} the case, however, that $G-G = 0$ for all $G$, although $G-G$ is always a zero game.  These easy results are important enough that we state and prove them formally.

\begin{lemma}\label{lem:plus-minus}
For any games $G$ and $H$,
\begin{enumerate}
\item $G-G$ is a zero game.
\item Suppose $G\gapprox 0$.  Then $H\gapprox 0$ implies $G+H\gapprox 0$, and $H\not\lapprox 0$ implies $G+H \not\lapprox 0$.
\item Suppose $G\lapprox 0$.  Then $H\lapprox 0$ implies $G+H\lapprox 0$, and $H\not\gapprox 0$ implies $G+H \not\gapprox 0$.
\item
$-(G+H) = -G-H$.
\end{enumerate}
\end{lemma}

\begin{proof}
For (1.): Any first move in $G-G$ is either a move in $G$ or in $-G$.  The second player can then simply play the equivalent move in the other game ($-G$ or $G$, respectively).  This is called a \emph{mirroring strategy}, and it guarantees a win for the second player.  For example, if, say, Left moves first and chooses some $g\in G^L$, then the game position is now $g-G = g + (-G)$, and so Right responds with $-g\in (-G)^R$, resulting in the game position $g-g$.  An inductive argument now shows that Right wins using this strategy.

For (2\@.) with $G\gapprox 0$: First, suppose $H\gapprox 0$ and Right moves first in $G+H$.  Then Right is moving either in $G$ or in $H$.  Left then chooses her winning response in whichever game Right moved in.  Left can continue this strategy until she wins.  For example, if Right chooses $h\in H^R$, then the game position is now $G+h$.  Since $H\gapprox 0$, we must have $h\not\lapprox 0$, and so there exists some $h'\in h^L$ such that $h' \gapprox 0$.  Left responds with $h'$, resulting in the position $G+h'$.  An inductive argument again proves that Left can win, and thus $G+H\gapprox 0$.  Now suppose $H\not\lapprox 0$.  Then there is some $h\in H^L$ such that $h\gapprox 0$.  If Left moves first in $G+H$, she chooses this $h$, leaving the position $G+h$ for Right, who moves next.  By the previous argument $G+h\gapprox 0$, and so Left can win it, because Right is moving first.  Thus $G+H$ is a first-move win for Left, i.e., $G+H\not\lapprox 0$.

(3.\@) is the dual of (2.\@) and follows by applying (2.\@) to the games $-G$ and $-H$ (and using (4.)).

For (4.): By induction (with the inductive hypothesis used for the fourth equality),
\begin{align*}
-(G+H) &= -\{(G^L+H)\cup(G+H^L)\mid (G^R+H)\cup(G+H^R)\} \\
&= \{-((G^R+H)\cup(G+H^R))\mid -((G^L+H)\cup(G+H^L))\} \\
&= \{(-(G^R+H))\cup(-(G+H^R))\mid (-(G^L+H))\cup(-(G+H^L))\} \\
&= \{(-G^R-H)\cup(-G-H^R)\mid (-G^L-H)\cup(-G-H^L)\} \\
&= \{((-G)^L-H)\cup(-G+(-H)^L)\mid ((-G)^R-H)\cup(-G+(-H)^R)\} \\
&= -G-H\;.
\end{align*}
\end{proof}

The outcome $o(G)$ of a game $G$ is certainly the first question to be asked about $G$, but it leaves out a lot of other important information about $G$.  It does not determine, for example, the outcome when $G$ is added to a fixed game $X$.  That is, it may be that two games $G$ and $H$ have the same outcome, but $o(G+X) \ne o(H+X)$ for some game $X$.  Indeed, defining $2 \eqdf \{1|\}$, one can check that $o(1) = o(2) = \cL$, but we have $o(2-1) = \cL$ (left wins by choosing $1\in 2^L$ when she gets the chance), whereas we know already from Lemma~\ref{lem:plus-minus} that $o(1-1) = \cP$.

Behavior under addition leads us to a finer classification of games.

\begin{definition}
Let $G$ and $H$ be games.  We say that $G$ and $H$ are \emph{equivalent}, written $G\approx H$, iff $o(G+X) = o(H+X)$ for all games $X$.\footnote{In much of the literature, the overloaded equality symbol $=$ is used for game equivalence.  We avoid that practice here, preferring to reserve $=$ for set theoretic equality.  There are some important game properties that are not $\approx$-invariant.}
\end{definition}

It follows immediately from the definition that $\approx$ is an equivalence relation on games, and we call the equivalence classes \emph{game values}.  We let $\bbG$ denote the Class\footnote{We will start to capitalize words that describe proper classes.} of all game values.\footnote{Since each game value itself is a proper Class, we really cannot consider it as a member of anything.  A standard fix for this in set theory is to represent each game value $v$ by the \emph{set} of elements of $v$ with minimum rank, so $\bbG$ becomes the Class of all such sets.}  Letting $X$ be the endgame $0$ in the definition shows that equivalent games have the same outcome.  Using the associativity of $+$, we also get that $G\approx H$ implies $G+X\approx H+X$ for any game $X$.  Thus $+$ respects equivalence and naturally lifts to a commutative and associative Operation (also denoted $+$) on $\bbG$.

\bigskip

The remaining goal of this subsection is finish showing that $\tuple{\bbG,+,\lapprox}$ is a partially ordered abelian Group.  We have built up enough basic machinery that we can accomplish our goal in a direct, arithmetic way, without referring to players' strategies.

\begin{lemma}\label{lem:zero-id}
A game $G$ is a zero game if and only if $G+H \approx H$ for all games $H$.
\end{lemma}

\begin{proof}
(Only if): It suffices to show that $o(G+H) = o(H)$ for any $H$, for then, given any game $X$, we have $o(G+H+X) = o(H+X)$ by substituting $H+X$ for $H$, hence the lemma.  Now by Lemma~\ref{lem:plus-minus}(2.), we get that $H\gapprox 0$ implies $G+H \gapprox 0$, and conversely, $H\not\gapprox 0$ implies $G+H\not\gapprox 0$.  A symmetric argument using Lemma~\ref{lem:plus-minus}(3.\@) proves that $H\lapprox 0$ if and only if $G+H\lapprox 0$.  Combining these statements implies $o(H) = o(G+H)$ as desired.

(If:) Set $H \eqdf 0$, the endgame.  Then $G = G+0 \approx 0$, and so $o(G) = o(0) = \cP$.
\end{proof}

\begin{corollary}
A game $G$ is a zero game if and only if $G\approx 0$ (where $0$ is the endgame).
\end{corollary}

\begin{proof}
For the forward direction, set $H \eqdf 0$ in Lemma~\ref{lem:zero-id}.  For the reverse direction, add any $H$ to both sides of the equivalence $G\approx 0$, then use Lemma~\ref{lem:zero-id} again.
\end{proof}

Here is our promised Preorder on games.

\begin{definition}\label{def:ge}
Let $G$ and $H$ be games.  We write $G \lapprox H$ (or $H\gapprox G$) to mean $H - G \gapprox 0$ (equivalently, $G-H\lapprox 0$).  As usual, we write $G<H$ to mean $G\lapprox H$ and $H\not\lapprox G$.\footnote{We now have two ways of interpreting the expression ``$G\gapprox 0$'': one using Definition~\ref{def:ge-zero} directly and the other using Definition~\ref{def:ge} with $0$ being the endgame.  One readily checks that the two interpretations coincide.}
\end{definition}

You can interpret $G < H$ informally as meaning that $H$ is more preferable a position for Left than $G$, or that $G$ is more preferable for Right than $H$.  For example, if Left is ever faced with moving in position $G$, and (let us pretend) she had the option of replacing $G$ with $H$ beforehand, she always wants to do so.

\begin{proposition}
The $\lapprox$ Relation on games is reflexive and transitive.
\end{proposition}

\begin{proof}
Reflexivity follows immediately from Lemma~\ref{lem:plus-minus}(1.).  For transitivity, suppose $G$, $H$, and $J$ are games such that $G\lapprox H$ and $H\lapprox J$.  Then
\[ J-G \approx J + (-H + H) - G = (J-H) + (H-G) \gapprox 0\;. \]
The first equivalence is by Lemma~\ref{lem:zero-id} and the fact that $-H+H$ is a zero game by Lemma~\ref{lem:plus-minus}(1.).  The final statement is by Lemma~\ref{lem:plus-minus}(2.), noticing that $J-H$ and $H-G$ are both $\ge 0$.  Thus $G\lapprox J$.
\end{proof}

\begin{proposition}\label{prop:partial-order}
For any two games $G$ and $H$, \ $G\approx H$ if and only if $G-H$ is a zero game, if and only if $G\lapprox H$ and $G\gapprox H$.
\end{proposition}

\begin{proof}
The second ``if and only if'' follows straight from the definitions.

(First only if:) $G\approx H$ implies $G-H \approx H-H$, since $+$ is $\approx$-invariant.  Then by Lemma~\ref{lem:plus-minus}(1.), $o(G-H) = o(H-H) = \cP$, i.e., $G-H$ is a zero game.

(First if:) By Lemma~\ref{lem:zero-id} and the fact that $H-H$ is also a zero game, we have $G \approx G+(H-H) = (G-H)+H \approx H$.
\end{proof}

The last two propositions show that the binary Relation $\lapprox$ on games is a Preorder that induces a partial Order on $\bbG$.  Proposition~\ref{prop:partial-order} also gives a good working criterion for proving or disproving game equivalence---just check whether $G-H$ is a second player win---without having to quantify over all games.

\begin{proposition}
$\tuple{\bbG,+}$ is an abelian Group, where the identity element is the $\approx$-equivalence class of zero games, and inverses are obtained by the negation Operator on games.
\end{proposition}

\begin{proof}
We already know that $+$ is associative and commutative on $\bbG$ and that the zero games form the identity under $+$ (Lemma~\ref{lem:zero-id}).  All we have left to show is that the negation Operator on games is $\approx$-invariant, for then, Lemma~\ref{lem:plus-minus}(1.\@) implies that it acts as the group theoretic inverse on $\bbG$.  Now suppose $G\approx H$ for any games $G$ and $H$.  Then $G\lapprox H$ and $G\gapprox H$ by Proposition~\ref{prop:partial-order}, i.e., $G-H \lapprox 0$ and $G-H \gapprox 0$.  Since by Lemma~\ref{lem:plus-minus}(4.),
$-G-(-H) = H-G$, we also have $-G \lapprox -H$ and $-G\gapprox -H$, hence $-G\approx -H$ by Proposition~\ref{prop:partial-order}.
\end{proof}

Finally, $\lapprox$ is translation-invariant on $\bbG$, making it a partially ordered abelian Group:

\begin{corollary}
For any games $G$, $H$, and $X$, if $G\lapprox H$ then $G+X \lapprox H+X$.
\end{corollary}

\begin{proof}
We have
\begin{align*}
G\lapprox H &\implies H-G\gapprox 0 \implies H-G+X-X \gapprox 0 \\
&\implies (H+X)-(G+X)\gapprox 0 \implies G+X \lapprox H+X\;.
\end{align*}
The first and last implications are by definition, and the other two are by Lemma~\ref{lem:plus-minus}.
\end{proof}

We next look at two important subclasses of games---the numeric games and the impartial games.

\subsection{Numeric games}
\label{sec:numeric}

A numeric game is one where at each position all the left options are $<$ all the right options.  Formally,

\begin{definition}
A game $G$ is \emph{numeric} iff $\ell < r$ for every $\ell\in G^L$ and $r\in G^R$, and further, every option of $G$ is numeric.
\end{definition}

One can show that $G$ is numeric if and only if $\ell < G$ for every $\ell\in G^L$ and $G<r$ for every $r\in G^R$.  If $H$ is also numeric, then either $G\lapprox H$ or $H\lapprox G$.  The $+$ and $-$ operations also yield numeric games when applied to numeric games.\footnote{The property of being numeric is \emph{not} invariant under $\approx$.  One can easily concoct two equivalent games, one of which is numeric and the other not.}  Numeric games have a peculiar property: making a move only worsens your position (for Left this means having to choose a smaller game; for Right, having to choose a larger game).  Thus neither player wants to make a move---if they were given the option to skip a turn, they would always take it.  For these games, an optimal play is easy to describe: Left always chooses a maximum left option (i.e., one that does the least damage), and Right always chooses a minimum right option, assuming these options exist.\footnote{In general, Left can win by choosing any option $\ell\gapprox 0$, and Right can win by choosing any option $r\lapprox 0$.}  This intuitive idea is formalized in the following theorem, which is referred to in the literature as the ``dominating rule.''  It applies to all games, not just numeric games.

\begin{theorem}
Let $G$ be a game.  If $y \lapprox \ell$ for some $\ell \in G^L$, then $G \approx \{y, G^L|G^R\}$.  Similarly, if $y \gapprox r$ for some $r \in G^R$, then $G \approx  \{G^L|G^R,y\}$.
\end{theorem}

If $y \lapprox \ell \in G^R$, then we say that $y$ is \emph{dominated} by $\ell$ in $G$.  Similarly, if $y\gapprox r \in G^R$, then $y$ is \emph{dominated} by $r$ in $G$.  We obtain equivalent games by removing dominated options.  A player never needs to play a dominated option; it is just as well (or better) to choose an option that dominates it.

Numeric games are called such because their values act like real numbers; for one thing, their values are totally ordered by $\lapprox$.  These games are constructed in a way somewhat akin to how the real numbers are constructed from the rationals via Dedekind cuts.  The left options of a game form the left cut, the right options the right cut, and the game itself represents a number strictly between the two.  The differences are that the two cuts might be bounded away from each other (one or the other may even be empty), and the left cut might contain a maximum element.

\subsubsection{Finite numeric games}

The values of \emph{finite} numeric games form a subgroup of $\bbG$ naturally isomorphic (in an order-preserving way) to the dyadic rational numbers under addition, according to the following ``simplicity rule'':

\begin{definition}
Let $G$ be a finite numeric game.  The \emph{(numerical) value} of $G$, denoted $v(G)$, is the unique rational number $a/2^k$ such that
\begin{enumerate}
\item
$k$ is the least nonnegative integer such that there exists an integer $a$ such that $v(\ell) < a/2^k$ for all $\ell\in G^L$ and $a/2^k < v(r)$ for all $r\in G^R$, and
\item
$a$ is the integer with the least absolute value satisfying (1.\@) above.
\end{enumerate}
\end{definition}

So for example, the endgame $0$ has value $v(0) = 0$, the game $1$ has value $v(1) = 1$, and the game $-1$ has value $v(-1) = -1$, as the notation suggests.  Intuitively, $|v(G)|$ indicates the number of ``free moves'' one of the players has before losing (Left if $v(G)>0$, and Right if $v(G)<0$).  In fact, for any two finite numeric games $P$ and $Q$, one can show that $v(P+Q) = v(P) + v(Q)$ and that $v(-P) = -v(P)$.  Also, $P\lapprox Q$ if and only if $v(P) \le v(Q)$.\footnote{One can define a purely game-theoretic multiplication operation on numeric games in such a way that $v(PQ) = v(P)v(Q)$ for all $P$ and $Q$.  See \cite{Conway:ONAG} for details.}  The valuation map $v$ is not one-to-one on games, but induces a one-to-one map on \emph{values} of numeric games.
  
To illustrate the simplicity rule, consider the game $h \eqdf \{0|1\}$.  The rule says that $v(h)$ is the simplest dyadic rational number strictly between $0$ and $1$, namely, $1/2$.  First note that Left can always win $h$ whether or not she plays first, so $h>0$.  If $v$ respects $+$, then we should also have $h+h \approx 1$.  Let us check this.  First consider $1-h$:
\begin{align*}
1-h &= 1+(-h) = \{0|\}+\{-1|0\} = \{0-h,1-1|1+0\} \\
&= \{-h,0|1\} \approx \{0|1\} = h
\end{align*}
(the equivalence is by the dominating rule and $-h<0$).  Thus
\[ h+h \approx h+(1-h) \approx 1\;. \]

Black-white poset games are numeric \cite{FGMST:black-white}.  Here we identify Black with Left and White with Right.  So for example, an antichain of $k$ black points has numeric value $k$, and an antichain of $k$ white nodes has numeric value $-k$.  Figure~\ref{fig:sample-bw-games} shows the numeric value of two simple, two-level black-white poset games.
\begin{figure}
\begin{center}
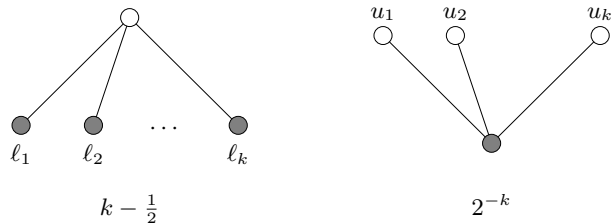
\caption{The numerical values of two simple black-white poset games.  The left has value $k-\frac{1}{2}$ and the right has value $2^{-k}$, for $k\ge 1$.}\label{fig:sample-bw-games}
\end{center}
\end{figure}

\begin{exercise}\label{ex:sample-bw-games}
Use the simplicity rule to prove the values in Figure~\ref{fig:sample-bw-games}.
\end{exercise}

The numerical values of arbitrary numeric games (not necessarily finite) form an ordered, real-closed field $\No$ into which the real numbers embed, but which also contains all the ordinals as well as infinitesimals \cite{Conway:ONAG}.  Donald Knuth dubbed $\No$ the \emph{surreal numbers} \cite{Knuth:surreal}, and they are formed via a transfinite construction.  The dyadic rationals are those constructed at finite stages, but numbers constructed through stage~$\omega$ already form a proper superset of $\reals$.

\subsection{Impartial games and Sprague-Grundy theory}
\label{sec:impartial}

A game is \emph{impartial} if at every position, the two players have the same options.  Formally,

\begin{definition}
A game $G$ is \emph{impartial} iff $G^L = G^R$ and every $g\in G^L$ is impartial.
\end{definition}

Equivalently, $G$ is impartial if and only if $G = -G$.  This means that values of impartial games are those that have order two in the group $\tuple{\bbG,+}$.

Examples of impartial games include $0$ and $*$.  Families of impartial games include \Nim, \Geo, \NK, and poset games.\footnote{Impartiality is not $\approx$-invariant.}  There is a beautiful theory of impartial games, developed by R. P. Sprague and P. M. Grundy \cite{Sprague:games,Grundy:games} that predates the more general theory of combinatorial games described in \cite{BCG:WW,Conway:ONAG}.  We develop the basics of this older theory here.  First note that, since there are no Left/Right biases, all impartial games are either zero ($\cP$) or fuzzy ($\cN$), and we can assume that Left always moves first.  We will call impartial zero games \emph{$\forall$-games} (``for all first moves \ldots'') and impartial fuzzy games \emph{$\exists$-games} (``there exists a first move such that \ldots'').  In this section only, we restrict our attention to impartial games, so when we say ``game,'' we mean impartial game.

Two (impartial) games $G$ and $H$ are equivalent ($G\approx H$) if and only if $G+H$ is a $\forall$-game, because $H = -H$ (Sprague and Grundy defined this notion for impartial games).  Applied to poset games, we get Lemma~\ref{lem:dual-symmetry} below (a partial generalization of Fact~\ref{fact:dual-symmetry}), which has been handy in finding the outcomes of some poset games.  A \emph{down set} in a partial order $P$ is a subset $S\subseteq P$ that is closed downwards under $\le$, i.e., $x\in S$ and $y\le x$ implies $y\in S$.

\begin{lemma}\label{lem:dual-symmetry}
Let $P$ be a poset and let $\map{\p}{P}{P}$ be such that $\p\circ\p = \identity_P$ and $x\le y \iff \p(x) \le \p(y)$ for all $x,y\in P$.  Let $F \eqdf \{x\in P \mid \p(x) = x\}$ be the set of fixed points of $\p$, considered as an induced subposet of $P$.  If $F$ is a down set, then $G\approx F$ as games.
\end{lemma}

\begin{proof}
Let $F'$ be a copy of $F$, disjoint from $P$, and consider the parallel union $P+F'$ as a poset game.  By Proposition~\ref{prop:partial-order}, we only need to show that $P+F'$ is a $\forall$-game, which we do by giving a winning strategy for the second player.  If the first player plays in $F$ or $F'$, then the second player plays the corresponding point in $F'$ or $F$, respectively.  If the first player plays some point $x\in G\setminus F$, then the second player responds by playing $\p(x)$.  Since $F$ is a down set, this latter pair of moves does not disturb $F$ or $F'$, and the resulting position in either case is seen to have the same basic form as the original game.
\end{proof}

One can associate an ordinal number with each game, which we call the \emph{g-number}\footnote{also called the \emph{Grundy number} or the \emph{NIM number}---not to be confused with the value of a numerical game} of the game, such that two games are equivalent if and only if they have the same g-number.  The g-number of a finite game is a natural number.  We will restrict ourselves to finite games.

\begin{definition}
Let $A$ be any coinfinite subset of $\nums$.  Define $\mex A$ (the \emph{\underline{m}inimum \underline{ex}cluded element} from $A$) to be the least natural number not in $A$, i.e.,
\[ \mex A \eqdf \min(\nums - A)\;. \]
More generally, for $i = 0,1,2,\ldots\,$, inductively define
\[ \mex_i A \eqdf \min\left(\nums - (A\cup\{\mex_0(A),\ldots,\mex_{i-1}A\})\right)\;, \]
the $i$'th least natural number not in $A$.  (So in particular, $\mex_0 A = \mex A$.)
\end{definition}

\begin{definition}
Let $G$ be any (finite) game.  Define the \emph{g-number} of $G$ as
\[ \gnum(G) \eqdf \mex{\gset(G)}\;, \]
where $\gset(G) \eqdf \{ \gnum(x) \mid x\in G^L \}$ is called the \emph{g-set} of $G$.
\end{definition}

That is, $\gnum(G)$ is the least natural number that is not the g-number of any option of $G$, and the set of g-numbers of options of $G$ is $\gset(G)$.  For example, $\gset(0) = \emptyset$, and so $\gnum(0) = 0$.  Also, $\gset(*) = \{\gnum(0)\} = \{0\}$, and so $\gnum(*) = 1$.

\begin{exercise}\label{ex:g-numbers}
Prove the following for any finite poset $P$ and any $n\in\nums$.
\begin{enumerate}
\item
$\gnum(P) \le \card P$.  (Generally, $\gnum(G) \le \card{G^L}$ for any impartial $G$.)
\item
$\gnum(C_n) = n$ for all $n\in\nums$.
\item\label{item:antichain}
$\gnum(A_n) = n\bmod 2$.
\item
$\gnum(V_n) = (n\bmod 2)+1$.
\end{enumerate}
What is $\gnum(\Lambda_n)$?  What is $\gnum(\Diamond_n)$?
\end{exercise}

\begin{exercise}
Describe $g(A_m/A_n)$ simply in terms of $m$ and $n$.
\end{exercise}

Here is the connection between the g-number and the outcome of a game.

\begin{proposition}\label{prop:g-eq-0}
A game $G$ is a $\forall$-game if and only if $\gnum(G) = 0$.
\end{proposition}

\begin{proof}[Proof idea]
If $\gnum(G)\ne 0$, then there is some option $x$ of $G$ that Left can play such that $\gnum(x) = 0$, but if $\gnum(G) = 0$, then no move Left makes can keep the g-number at $0$.
\end{proof}

The central theorem of Sprague-Grundy theory---an amazing theorem with a completely nonintuitive proof---concerns the g-number of the sum of two games.

\begin{definition}
For any $m,n\in\nums$, define $m\oplus n$ to be the natural number $k$ whose binary representation is the bitwise exclusive OR of the binary representations of $m$ and $n$.  We may also call $k$ the \emph{bitwise XOR} of $m$ and $n$.
\end{definition}

For example, $23 \oplus 13 = 10111 \oplus 01101 = 11010 = 26$.

\begin{theorem}[Sprague, Grundy \cite{Sprague:games,Grundy:games}]\label{thm:sg}
For any finite games $G$ and $H$,
\[ \gnum(G+H) = \gnum(G)\oplus\gnum(H)\;. \]
\end{theorem}

\begin{proof}
As with most of these proofs, we use induction.  Let $G$ and $H$ be games.  If Left plays some $x\in G^L$, say, then $\gnum(x) \ne \gnum(G)$, and so
\begin{align*}
\gnum(x + H) &= \gnum(x) \oplus \gnum(H) & \mbox{(inductive hypothesis)} \\
&\ne \gnum(G) \oplus \gnum(H) & \mbox{(because $\gnum(G)\ne \gnum(x)$.)}
\end{align*}
Similarly, $\gnum(G + y) \ne \gnum(G) \oplus \gnum(H)$ for any $y\in H^L$.  This means that $\gnum(G)\oplus\gnum(H)$ is not the g-number of any option of $G+H$.  We'll be done if we can show that every natural number \emph{less than} $\gnum(G)\oplus\gnum(H)$ \emph{is} the g-number of some option of $G+H$.

Set $g_G \eqdf \gnum(G)$ and $g_H \eqdf \gnum(H)$, and let $m = g_G \oplus g_H$.  Fix any $k<m$.  We find an option of $G+H$ with g-number $k$.  Let's assign numbers to bit positions, $0$ being the least significant, $1$ being the second least, and so forth.  For any number $\ell\in\nums$, let $(\ell)_i$ be the $i$th least significant bit of $\ell$ (starting with $\ell_0$).  Since $k<m$, there exists a unique $i$ such that $(k)_i = 0$, $(m)_i = 1$, and $(k)_j = (m)_j$ for all $j>i$.  Fix this $i$.  We have $(g_G)_i \oplus (g_H)_i = (m)_i = 1$, and so one of $g_G$ and $g_H$ has a $1$ in the $i$th position and the other a $0$.  Suppose first that $g_G$ has a $1$ in that position.  Then Left can play in $G$ to ``clear'' that bit: First, notice that $k \oplus g_H < g_G$.  Why?  Because
\[ (k\oplus g_H)_i = (k)_i \oplus (g_H)_i = 0\oplus 0 = 0 < 1 = (g_G)_i\;, \]
and for all $j>i$,
\[ (k\oplus g_H)_j = (k)_j \oplus (g_H)_j = (m)_j \oplus (g_H)_j = (g_G)_j \oplus (g_H)_j \oplus (g_H)_j = (g_G)_j\;. \]
So there must exist an $x\in G^L$ such that $\gnum(x) = k \oplus g_H$, and then by the inductive hypothesis,
\[ \gnum(x + H) = \gnum(x) \oplus g_H = k \oplus g_H \oplus g_H = k\;. \]
Similarly, if $(g_H)_i = 1$ and $(g_G)_i = 0$, then there exists $y\in H^L$ such that $\gnum(P + y) = k$.
\end{proof}

\begin{corollary}\label{cor:g-number-equiv}
Two impartial games $G$ and $H$ are equivalent if and only if $\gnum(G) = \gnum(H)$.
\end{corollary}

\begin{proof}
$G$ and $H$ are equivalent iff $G+H$ is a $\forall$-game, iff $\gnum(G+H) = 0$ (Proposition~\ref{prop:g-eq-0}), iff $\gnum(G)\oplus \gnum(H) = 0$ (Theorem~\ref{thm:sg}), iff $\gnum(G) = \gnum(H)$.
\end{proof}

Since every natural number $n$ is the g-number of the poset game $C_n$, this means that every game is equivalent to a single NIM stack.

We can use Theorem~\ref{thm:sg} to solve \Nim.  Given a \Nim\ game $P = C_{n_1} + \cdots + C_{n_k}$, we get $\gnum(P) = n_1\oplus \cdots \oplus n_k$.  If this number is nonzero, then let $i$ be largest such that $(\gnum(P))_i = 1$.  Alice can win by choosing a $j$ such that $(n_j)_i = 1$ and playing in $C_{n_j}$ to reduce its length (and hence its g-number) from $n_j$ to $n_j\oplus (\gnum(P))_i$.  This makes the g-number of the whole \Nim\ game zero.

We can use Corollary~\ref{cor:g-number-equiv} and Lemma~\ref{lem:dual-symmetry} to find the g-numbers of some natural, interesting posets.  We give Proposition~\ref{prop:level-sets} below as an example.  For positive integer $n$, let $[n] \eqdf \{1,2,\ldots,n\}$, and let $2^{[n]}$ be the powerset of $[n]$, partially ordered by $\subseteq$.  For $0\le k\le n$, we let $\binom{[n]}{k}\subseteq 2^{[n]}$ be the set of all $k$-element subsets of $[n]$.  Then we have the following:

\begin{proposition}\label{prop:level-sets}
Let $n>0$ be even and let $1\le k < k' \le n$ be such that $k'$ is odd.  Let $n = n_{j-1}\cdots n_1n_0$ and $k = k_{j-1}\cdots k_1 k_0$ be binary representations of $n$ and $k$, respectively, where $n_i,k_i\in\{0,1\}$ for $0\le i<j$.  Letting $P \eqdf \binom{[n]}{k} \cup \binom{[n]}{k'}$,
we have
\[ g(P) = \left\{\begin{array}{ll}
1 & \mbox{if $k_i > n_i$ for some $0\le i<j$,} \\
0 & \mbox{otherwise.}
\end{array}\right. \]
In particular, if $k$ is even, then $g(P) = \binom{n/2}{k/2} \bmod 2$.
\end{proposition}

\begin{proof}
For sets $A$ and $B$, we say that \emph{$A$ respects $B$} if either $B\subseteq A$ or $A\cap B = \emptyset$.  Define the map $\map{\p}{[n]}{[n]}$ so that $\p(2i) = 2i-1$ and $\p(2i-1) = 2i$, for all $1\le i \le n/2$.  Then $\p$ swaps the elements of each two-element set $s_i \eqdf \{2i-1,2i\}$.  We lift the involution $\p$ to an involution $\map{\p'}{2^{[n]}}{2^{[n]}}$ in the usual way: $\p'(S) \eqdf \{\p(x)\mid x\in S\}$ for all $S\subseteq [n]$.  Notice that $\p'$ preserves set cardinality, and so $\p'$ maps $P$ onto $P$.  Also notice that $\p'(S) = S$ if and only if $S$ respects all the $s_i$.

Let $F$ be the set of all fixed points of $\p'$.  Since $k'$ is odd, no $S\in\binom{[n]}{k'}$ can respect all the $s_i$, and thus $\p'(S) \ne S$ for all $S\in\binom{[n]}{k'}$.  It follows immediately that $F\subseteq\binom{[n]}{k}$ is a down set, and so we have $g(P) = g(F)$ by Lemma~\ref{lem:dual-symmetry} and Corollary~\ref{cor:g-number-equiv}.  Since $F$ is also an antichain, we have $g(F) = \card{F} \bmod 2$ (Exercise~\ref{ex:g-numbers}(\ref{item:antichain})).  Now $F$ consists of those $k$-sets that respect all the $s_i$.  If $k$ is odd, then $F=\emptyset$, whence $0 = g(F) = g(P)$, and we also have $1=k_0>n_0=0$ so the proposition holds.  If $k$ is even, then by a simple combinatorial argument we have $\card F = \binom{n/2}{k/2}$---by selecting exactly $k/2$ of the $s_i$ to be included in each element of $F$.  Therefore, we have $g(P) = g(F) = \card F \bmod 2 = \binom{n/2}{k/2} \bmod 2$, and the proposition follows by Lucas's theorem.
\end{proof}

Proposition~\ref{prop:level-sets} clearly still holds if we include in $P$ any number of odd levels of $2^{[n]}$ above the $k$th level (including zero).

\bigskip

Theorem~\ref{thm:sg} shows how the $g$-number behaves under parallel unions of posets (Definition~\ref{def:series-parallel}).  How does the g-number behave under series unions?  Unfortunately, $\gnum(P/Q)$ might not depend solely on $\gnum(P)$ and $\gnum(Q)$.  For example, $g(V_2) = g(C_1) = 1$, but $g(C_1/V_2) = g(\Diamond_2) = 3$ whereas $g(C_1/C_1) = g(C_2) = 2$.
However, $\gset(P/Q)$ \emph{does} depend solely on $\gset(P)$ and $\gset(Q)$ for any posets $P$ and $Q$, and this fact forms the basis of the Deuber \& Thomass\'e algorithm of the next section.

There is one important case where $\gnum(P/Q)$ does only depend on $\gnum(P)$ and $\gnum(Q)$:

\begin{fact}\label{fact:series-gnum}
For any finite poset $P$ and any $k\ge 0$,
\[ \gnum\left(\frac{P}{C_k}\right) = \gnum(P) + k\;. \]
\end{fact}

This can shown by first showing that $\gnum(P/C_1) = \gnum(P)+1$, then using induction on $k$. By Fact~\ref{fact:series-gnum}, we get that $g(\Diamond_n) = 1+g(\Lambda_n)$ for example.

\section{Upper bounds}
\label{sec:ub}

When asking about the computational difficulty of determining the outcome of a game, we really mean a family of similar games, represented in some way as finite inputs.  In discussing game complexity, we will abuse terminology and refer to a family of games simply as a game.  (The same abuse occurs in other areas of complexity, notably circuit complexity.)  We will also use the same small-caps notation to refer both to a family of games and to the corresponding decision problem about the outcomes.

Perhaps the most common upper bound in the literature on the complexity of a game is membership in $\PSPACE$.  Without pursuing it further, we will just mention that, if a game $G$ of size $n$ satisfies: (i) every position of $G$ has size polynomial in $n$; (ii) the length of any play of $G$ is polynomial in $n$; and (iii) there are polynomial-time (or even just polynomial-space) algorithms computing the ``left option of'' and ``right option of'' relations on the positions of $G$, then $o(G)$ can be computed in polynomial space.  These properties are shared by many, many games.

In this section we will give some better upper bounds on some classes of finite poset games, the best one being that N-free poset games are in $\PTIME$ \cite{DT:N-free}.  We will assume that a poset is represented by its Hasse diagram, a directed acyclic graph (DAG) in which each element is represented as a node and an arc is placed from a node for element $x$ to the node for $y$ when $x<y$ and there is no element $z$ such that $x<z<y$.  The poset is the reflexive, transitive closure of the edge relation of the DAG.

\subsection{N-free games}

With the Hasse diagram representation, we can apply results from graph theory to devise efficient ways to calculate Grundy numbers for certain classes of games.  A good example is the class of N-free poset games.  An ``N'' in a poset is a set of four elements $\{a, b, c, d\}$ such that $a<b$, $c<d$, $c<b$, and the three other pairs are incomparable.  When drawn as a Hasse diagram the arcs indicating comparability form the letter ``N''.  A poset is \emph{N-free} if it contains no N as an induced subposet.  We let \Nfree\ denote the class of N-free poset games.

Valdes, Tarjan, and Lawler \cite{VTL:series-parallel} show that an N-free DAG can be constructed in linear time from a set of single nodes.  New components are created either by applying parallel union ($G+H$) or by applying series union ($G/H$).  As with posets, the parallel union is the disjoint union of $G$ and $H$.  The series union is a single DAG formed by giving to every element in $H$ with out-degree 0 (the sinks in $H$) an arc to every element in $G$ with in-degree 0 (the sources in $G$).  This gives the Hasse diagram of the series union of the corresponding posets.  Their algorithm provides a sequence of $+$ and $/$ operations that will construct a given N-free DAG from single points.

Deuber \& Thomass\'e \cite{DT:N-free} show that $\Nfree \in \PTIME$ by applying this construction to demonstrate how to calculate the g-number of an N-free poset game based on the sequence of construction steps obtained by the VTL algorithm above.  Their algorithm, which we now describe, works by keeping track of the g-sets of the posets obtained in the intermediate steps of the construction, rather than the g-numbers.  There is no need to store the g-numbers, because the g-number of any poset can always be easily computed from its g-set by taking the $\mex$.

The g-number of a single node is $1$.  This is the base case.

\begin{fact}\label{fact:parallel-g-sets}
Given posets $P$ and $Q$, the g-set of the parallel union $P+Q$ is
\begin{align*}
\gset(P+Q) &= \{ \gnum(P+Q_q): q\in Q\} \cup \{ \gnum(P_p+Q): p\in P\} \\
&= \{\gnum(P) \oplus \gnum(Q_q): q \in Q\} \cup \{\gnum(P_p) \oplus \gnum(Q): p \in P\}\;.
\end{align*}
\end{fact}

The second equality follows from the Sprague-Grundy theorem.  This is easy to see if you consider the root of the game tree for $P+Q$.  Each of its children results from playing either an element in $P$ or one in $Q$.  The left-hand set in the union contains the g-numbers of the games resulting from playing an element in $Q$; the right-hand set from playing an element in $P$.  Their union is the g-set of $P+Q$, so its g-number is the $\mex$ of that set.  

To calculate the g-set of a series union, we will need the definition of the \emph{Grundy product} of two finite sets of natural numbers:
\[ A\odot B \eqdf B \cup \{\mex_a B \mid a \in A\}\;. \]
$A\odot B$ is again a finite set of natural numbers that is easy to compute given $A$ and $B$.  Basically, $A\odot B$ unions $B$ with the version of $A$ we get after re-indexing the natural numbers to go ``around'' $B$.  Notice that $\mex(A\odot B) = \mex_{\mex A}B$.  We will use this fact below.

\begin{lemma}[Deuber \& Thomass\'e~\cite{DT:N-free}]\label{lem:series-g-sets}
For any finite posets $P$ and $Q$, \ $\gset(P/Q) = \gset(P) \odot \gset(Q) = \gset(Q) \union \{\mex_i(\gset(Q)): i \in \gset(P)\}$.
\end{lemma}

The left-hand set of the union results from playing an element in $Q$, which removes all of the elements in $P$.  Using induction, we can see what happens when an element in $P$ is played.

\begin{proof}[Proof of Lemma~\ref{lem:series-g-sets}]
The fifth equality uses the inductive hypothesis.
\begin{align*}
\gset(P/Q) & = \{\gnum((P/Q)_r): r \in P/Q\} \\
           & = \{\gnum((P/Q)_p): p \in P\} \cup \{\gnum((P/Q)_q): q \in Q\} \\
           & = \{\gnum((P_p/Q)): p \in P\} \cup  \{\gnum(Q_q): q \in Q\} \\
           & = \{\mex \gset(P_p/Q): p\in P\} \cup \gset(Q) \\
           & = \{\mex(\gset(P_p) \odot \gset(Q)): p \in P\} \cup \gset(Q) \\
           & = \{\mex_{\mex\gset(P_p)}(\gset(Q)): p \in P\} \cup \gset(Q) \\
           & = \{\mex_{\gnum(P_p)}(\gset(Q)): p \in P\} \cup \gset(Q) \\
           & = \{\mex_i(\gset(Q)): i \in \gset(P)\} \cup \gset(Q) \\
           & = \gset(P) \odot \gset(Q) 
\end{align*}
\end{proof}

In particular, the g-number of $P/Q$ is greater than or equal to the sum of the g-numbers of $P$ and $Q$.  Notably, it's an equality if $Q$ is $C_n$ for some $n$ (Fact~\ref{fact:series-gnum}) and the reason is that the g-set of $C_n$ has no gaps, that is, it contains all of the values from 0 to $n-1$.  It's easy to see that it's true when $P$ and $Q$ are both singletons.  Their g-numbers are both 1 and forming their series-union creates a NIM stack of size 2 and that has g-number 2.  

Another way to understand Lemma~\ref{lem:series-g-sets} is to consider the game tree of $P/Q$, and we'll look at the simple case where $P$ is an arbitrary game with g-number $k$ and $Q$ is a singleton.  Consider the root node $r$ of the game tree of $P/Q$.  One of its children represents playing the single element in $Q$ and that child has g-number $0$.  The rest of $r$'s children represent game configurations reached by playing an element in $P$.  By the induction hypothesis the g-number of each of these nodes will be one more than in $P$'s game tree where they had g-numbers $0$ to $k-1$, and perhaps g-numbers $k+1$ and larger.  So in $P/Q$'s tree they have g-numbers $1$ to $k$, with perhaps g-numbers $k+2$ or larger.  Because the child reached by playing $Q$'s single element has g-number $0$, the first missing value in the g-set formed from these g-numbers is $k+1$.  

Now using Fact~\ref{fact:parallel-g-sets} and Lemma~\ref{lem:series-g-sets}, the decomposition described in \cite{VTL:series-parallel} can generate a binary tree where each internal node is labeled with a poset $P$ and an operation (parallel union or series union), and its children are the two posets combined to form $P$.  Starting with each leaf, where the poset is a singleton and the g-set is $\{0\}$, and moving up the tree, one can apply Fact~\ref{fact:parallel-g-sets} and Lemma~\ref{lem:series-g-sets} to compute the g-set of the root (and none of the g-numbers involved exceed the size of the final poset).  This can all be done in time $O(n^4)$.

\subsection{Results on some classes of games with N's}

General results for classes of games containing an ``N'' have been few.  In 2003, Steven Byrnes \cite{Byrnes:chomp} proved a poset game periodicity theorem, which applies to, among others, \Chomp-like games, which contain many ``N''-configurations.  

Here's the theorem, essentially as stated in the paper:
\begin{theorem}
In an infinite poset game $X$, suppose we have two infinite chains $C$ ($c_1 < c_2 < \cdots$) and $D$ ($d_1 < d_2 < \cdots$), and a finite subset $A$, all pairwise disjoint, and assume that no element of $C$ is less than an element of $D$. Let $A_{m,n} = A \cup C \cup D - \{x \in X | x \ge c_{m+1}\}-\{x \in X | x \ge d_{n+1}\}$
(that is, $A_{m,n}$ is the position that results from starting with the poset $A \cup C \cup D$, then making the two moves $c_{m+1}$ and $d_{n+1}$). Let $k$ be a nonnegative integer. Then either:
\begin{enumerate}
\item there are only finitely many different $A_{m,n}$ with g-number $k$; or
\item we can find a positive integer $p$ such that, for large enough $n$, $g(A_{m,n}) = k$ if and only if $g(A_{m+p,n+p}) = k$.
\end{enumerate}
Thus, as the poset $A$ expands along the chains $C$ and $D$, positions with any fixed g-number have a regular structure.
\end{theorem}

A simple example of a class of games covered by the theorem is the family of two-stack \Nim\ games, where $A$ is empty and $A_{m,n}$ consists of an $m$-chain and an $n$-chain.  The g-number $0$ occurs for every $A_{n,n}$ so the periodicity is $1$.  The g-number $1$ occurs for every $A_{2n, 2n+1}$ and so has periodicity $2$.  In fact, one can find a periodic repetition for every g-number.  The surprising thing is that this is still true when you allow elements in one chain to be less than elements in the other.

Another family contains \Chomp, described in Section~\ref{sec:poset-game-examples}.  We can generalize \Chomp\ to games where the rows do not have to contain the same number of elements.  Byrnes showed that for such games there is a periodicity in the g-numbers when we fix the size of all but the top two rows.

As Byrnes claims, this yields a polynomial-time decision algorithm for each family generated from a fixed $A$ but not a uniformly polynomial-time algorithm across the families, as the time is parameterized by $A$.

\subsubsection{Bounded-width poset games}
\label{subsubsec:bwpg}

If a poset $P$ has width $k$, that is, if $k$ is the maximum size of any antichain in $P$, then there are only $|P|^k$ many positions at most in the game: if $x_0,x_1,\ldots,x_{n-1} \in P$ are the elements chosen by the players in the first $n$ moves of the game, then the resulting position is completely determined by the minimal elements of the set $\{x_0,\ldots,x_{n-1}\}$, i.e., an antichain of size $\le k$.

This means that, for constant $k$, one can compute the g-number of $P$ in polynomial time using dynamic programming.  The exponent on the running time depends on $k$, however.  For certain families of bounded-width posets, one can beat the time of the dynamic programming algorithm; for example, one can compute the g-number of width-2 games in linear time.

\subsubsection{Parity-uniform poset games}

Daniel Grier recently showed that computing arbitrary poset game outcomes is $\PSPACE$-complete (Theorem~\ref{thm:poset-pspace-complete} and its proof, below).  He reduces from True Quantified Boolean Formulas (see Section~\ref{sec:pspace-hard}).  His reduction constructs posets with only three levels, i.e., posets that can be partitioned into three antichains (equivalently, the maximum size of a chain is $3$).  An obvious follow-up question is whether two-level poset games remain $\PSPACE$-complete.  This question is still open, but in \cite{FGKT:two-level} it is shown that a certain subclass of two-level posets is easy, namely, the ``parity-uniform'' posets.  This result builds on and extends earlier results of Fraenkel \& Scheinerman~\cite{FS:hypergraph-game}.

\begin{definition}[\cite{FGKT:two-level}]\label{def:parity-uniform}
Let $P$ be a two-level poset, partitioned into two sets $T$ (top points) and $B$ (bottom points) so that for any $x,y\in P$, if $x<y$ then $x\in B$ and $y\in T$.  We can then view $P$ as a bipartite graph, where the points of $P$ are the vertices and with an edge drawn between each $x$ and $y$ iff $x<y$.

We say that $P$ (viewed as a bipartite graph) is \emph{parity-uniform} iff: (i) all top points have the same degree parity (i.e., degrees of top points are either all even or all odd); and (ii) there is a bipartition of the bottom points such that every top point has a odd number of neighbors in at least one of the partitions (one of the partitions could also be empty).
\end{definition}

A parity-uniform poset has a simple expression for its g-number.

\begin{theorem}[F et al.~\cite{FGKT:two-level}]\label{thm:parity-uniform}
Let $P$ be a two-level poset, viewed as a bipartite graph with bipartition $T,B$ as in Definition~\ref{def:parity-uniform}, and suppose that $P$ is parity-uniform.  Let $p\in\{0,1\}$ be the common degree parity of the points in $T$.  Let $b \eqdf  \card B \bmod 2$ and let $t \eqdf \card T \bmod 2$.  Then
\[ \gnum(P) = b \oplus t(p\oplus 2)\;. \]
\end{theorem}

Theorem~\ref{thm:parity-uniform} is proved by induction on $\card P$ together with a case analysis.

\section{Lower bounds}
\label{sec:lb}

In this section we give some lower bounds on game complexity.  There is a vast literature on combinatorial game complexity, and we make no attempt to be thorough, but rather concentrate on poset game complexity.

\subsection{A note about representations of games}
\label{sec:representations}

The complexity of a game depends quite a bit on its representation.  The choice of representation is usually straightforward, but not always.  For example, how should we represent an N-free poset?  Just via its Hasse diagram, or via an expression for the poset in terms of single points and parallel union and series union operators?  The results of Valdes, et al\@.~\cite{VTL:series-parallel} show that one representation can be converted into the other in polynomial time, so the choice of representation is not an issue unless we want to consider complexity classes within $\PTIME$ or more succinct representations of posets, as we will do below.  There, fortunately, our hardness results apply to either representation.

Even if the representation of a game is clear, the results may be counterintuitive.  For example, how should we represent members of the class of \emph{all} finite games?  In Section~\ref{sec:definitions}, we defined a game as an ordered pair of its left and right options.  We must then represent the options, and the options of options, and so on.  In effect, to represent an arbitrary finite game explicitly, we must give its entire game tree (actually, game DAG, since different sequences of moves may end up in the same position).  Under this representation, there is a straightforward algorithm to compute the outcome of any game: use dynamic programming to find the outcome of every position in the game.  Since every position is encoded in the string representing the game, this algorithm runs in polynomial time.

What makes a game hard, then, is that we have a succinct representation for it that does not apply to all games.  For example, the obvious representation of a poset game is the poset itself, and the number of positions is typically exponential in the size of the poset.  Subfamilies of poset games may have even more succinct representations.  For example, a \Nim\ game can be represented as a finite list of natural numbers in binary, giving the sizes of the stacks, and a game of \Chomp\ can be represented with just two natural numbers $m$ and $n$ in binary, giving the dimensions of the grid.  Notice that this \Chomp\ representation is significantly shorter than what is needed to represent an arbitrary \emph{position} in a \Chomp\ game; the latter is polynomial in $m+n$.

In what sense does finding a winning strategy in \Chomp\ reduce to determining the outcome of \Chomp\ games?  We already know that every \Chomp\ game is an $\exists$-game because it has a maximal point.  We could find a winning strategy if we were able to determine the outcome of every \Chomp\ position, but even writing down a query to an ``outcome oracle'' takes time linear in $m+n$, which is exponential in the input size.  The more modest goal of finding a winning first move may be more feasible, because the position after one move is simple enough to describe by a polynomial-length query string.  To our knowledge, no efficient algorithm is known to determine the outcome of an arbitrary \Chomp\ position after a single move, even allowing time $(m+n)^{O(1)}$.

We will more to say about representations below when we discuss lower bounds for poset games within the complexity class $\PTIME$.

\subsection{Some PSPACE-hard games}
\label{sec:pspace-hard}

Many games have been shown $\PSPACE$-hard over the years.  Early on, Even and Tarjan showed that $\Hex$ generalized to arbitrary graphs is $\PSPACE$-complete \cite{ET:Hex}.  A typical proof of $\PSPACE$-hardness reduces the $\PSPACE$-complete True Quantified Boolean Formulas (TQBF \cite{Stockmeyer:PH}) problem to the outcome of a game.  We can consider a quantified Boolean formula $\p = (\exists x_1)(\forall x_2) \cdots \psi$ (where $\psi$ is a Boolean formula in conjunctive normal form (cnf)) itself as a game, where players alternate choosing truth values for $x_1,x_2,\ldots$, the first player (Right, say) winning if the resulting instantiation of $\psi$ is true, and Left winning otherwise.\footnote{This is technically not a combinatorial game by our definition, because the end condition is different.  One can modify the game slightly to make it fit our definition, however.}

TQBF seems ideal for encoding into other games.  Thomas Schaefer showed a number of interesting games to be $\PSPACE$-hard this way \cite{Schaefer:PSPACE-complete}.  One interesting variant of TQBF that Schaefer proved $\PSPACE$-complete is the game where a positive Boolean formula $\psi$ is in cnf with no negations, and players alternate choosing truth values for the Boolean variables.  Schaefer called this game $G_\textup{pos}(\mbox{POS CNF})$.  Unlike TQBF, however, the variables need not be chosen in order; players may choose to assign a truth value to any unassigned variable on any move.  Left (who moves first) wins if $\psi$ is true after all variables have been chosen, and Right wins otherwise.  Since $\psi$ is positive, Left always wants to set variables to $1$ and Right to $0$.

As another example, consider $\Geo$.  The input is a directed graph~$G$ and a designated vertex~$s$ of~$G$ on which a token initially rests.  The two players alternate moving the token on~$G$ from one node to a neighboring node, trying to force the opponent to move to a node that has already been visited.  $\Geo$ is a well-known $\PSPACE$-complete game~\cite{Schaefer:PSPACE-complete,Sipser:theory2}.  In \cite{LicSip80}, Lichtenstein \& Sipser show that $\Geo$ is $\PSPACE$-complete even for bipartite graphs.

An obvious way to turn $\Geo$ into a black-white game is to color the nodes of graph~$G$ black and white.  Each player is then only allowed to move the token to a node of their own color.  Since moves are allowed only to neighboring nodes, the black-white version is equivalent to the uncolored version on bipartite graphs.  The standard method of showing that $\Geo$ is $\PSPACE$-complete is via a reduction from True Quantified Boolean Formulas (TQBF) to $\Geo$ (see for example~\cite{Sipser:theory2}).  Observe that the graph constructed in this reduction is not bipartite.  That is, there are nodes that potentially may be played by both players.  Hence, we cannot directly conclude that the black-white version is $\PSPACE$-complete.  However, in \cite{LicSip80} Lichtenstein \& Sipser show that $\Geo$ is indeed $\PSPACE$-complete for bipartite graphs.

We now consider the game $\NK$.  This game is defined on an undirected graph~$G$.  The players alternately play an arbitrary node from~$G$.  In one move, playing node~$v$ removes~$v$ and all the direct neighbors of~$v$ from~$G$.  In the black-white version of the game, we color the nodes black and white.  Schaefer~\cite{Schaefer:PSPACE-complete} showed that determining the winner of an arbitrary $\NK$ instance is $\PSPACE$-complete.  He also extended the reduction to bipartite graphs, which automatically yields a reduction to the black-white version of the game (see~\cite{GJ:NPcomplete}).  Therefore, black-white $\NK$ is also $\PSPACE$-complete.

The game of \Col\ \cite{BCG:WW} is a two-player combinatorial strategy game played on a simple planar graph, some of whose vertices may be colored black or white.  During the game, the players alternate coloring the uncolored vertices of the graph.  One player colors vertices white and the other player colors vertices black.  A player is not allowed to color a vertex neighboring another vertex of the same color.  The first player unable to color a vertex loses.  A well-known theorem about \Col\ is that the value of any game is either $x$ or $x + *$ where $x$ is a number.  Removing the restriction that \Col\ games be played on planar graphs and considering only those games in which no vertex is already colored, we get a new game, \GenCol\ (generalized \Col).  It is shown in \cite{FGMST:black-white} that \GenCol\ is $\PSPACE$-complete; furthermore, \GenCol\ games only assume the two very simply game values $0$ and $*$.

In \cite{SC:provably-difficult}, Stockmeyer \& Chandra give examples of games that are complete for exponential time and thus provably infeasible.

\subsection{Lower bounds for poset games}

Until recently, virtually no hardness results were known relating to poset games, and the question of the complexity of determining the outcome of a game was wide open, save the easy observation that it is in $\PSPACE$.

For the moment, let $\PG$ informally denote the decision problem of determining the outcome of a arbitrary given (impartial) poset game, that is, whether or not the first player (Alice) can win the game with perfect play.  The first lower bound on the complexity of $\PG$ we are aware of, and it is a modest one, was proved by Fabian Wagner \cite{Wagner:nim} in 2009.  He showed that $\PG$ is $\LOGSPACE$-hard\footnote{$\LOGSPACE$ is short for LOGSPACE.} under $\FO$-reductions (First-Order reductions).  This is enough to show, for example, that $\PG\notin\AC^0$.  Soon after, Thomas Thierauf \cite{Thierauf:nim} showed that $\PG$ is hard for $\NL$ under $\AC^0$ reductions.\footnote{$\NL$ is nondeterministic LOGSPACE.}  A breakthrough came in 2010, when Adam Kalinich, then a high school student near Chicago, Illinois, showed that $\PG$ is hard for $\NC^1$ under $\AC^0$ reductions \cite{Kalinich:flip}.  For the proof, he invents a clever way to obliviously ``flip'' the outcome of a game, i.e., to toggle the outcome between $\exists$ and $\forall$.  This allows for the simulation of a NOT-gate in an $\NC^1$ circuit.  (An OR-gate can be simulated by the series union construction of Definition~\ref{def:series-parallel}.  See below.)

The astute reader will notice that Kalinich's result appears to be weaker than the other two earlier results.  In fact, the three results are actually incomparable with each other, because they make different assumptions about how poset games are represented as inputs.  We say more about this below, but first we mention that Wagner's and Thierauf's results both hold even when restricted to \Nim\ games with two stacks, and Kalinich's result holds restricted to $N$-free games.  Modest as they are, these are currently the best lower bound we know of for N-free poset games.

Very recently, the complexity of $\PG$ was settled completely by Daniel Grier, an undergraduate at the University of South Carolina \cite{Grier:poset-games}.  He showed that $\PG$ is $\PSPACE$-complete via a polynomial reduction (henceforth, p-reduction) from $\NK$.  Here, it is not important how a game is represented as an input, so long as the encoding is reasonable.  His proof shows that $\PSPACE$-completeness is still true when restricted to three-level games, i.e., posets where every chain has size at most three (equivalently, posets that are partitionable into at most three antichains).  The games used in the reduction are of course not N-free.


\subsection{Representing posets as input}

As we discussed above, for any of the various well-studied families of poset games (\Chomp, \Divisors, \Nim, etc.), there is usually an obvious and natural way to represent a game as input.  For example, an instance of \Chomp\ can be given with just two positive integers, one positive integer for Divisors, and a finite list of positive integers for \Nim, giving the heights of the stacks.  When considering arbitrary finite posets, however, there is no single natural way to represent a poset as input, but rather a handful of possibilities, and these may affect the complexity of various types of poset games.  We consider two broad genres of poset representation:
\begin{description}
\item[Explicit] The poset is represented by an explicit data structure, including the set of points and the relations between them.  In this representation, the size of the poset is always comparable to the size of the input.
\item[Succinct (Implicit)] The poset is represented by a Boolean circuit with two $n$-bit inputs.  The inputs to the circuit uniquely represent the points of the poset, and the ($1$-bit) output gives the binary relation between these two inputs.  In this representation, the size of the poset can be exponential in the size of the circuit.
\end{description}
Within each representational genre, we will consider three general approaches to encoding a poset $P$, in order from ``easiest to work with'' to ``hardest to work with'':
\begin{description}
\item[Partial Order (PO)]  $P$ is given as a reflexive, transitive, directed acyclic graph, where there is an edge from $x$ to $y$ iff $x\le y$.
\item[Hasse Diagram (HD)]  $P$ is given as a directed acyclic graph whose reflexive, transitive closure (i.e., reachability relation) is the ordering $\le$.  The graph then gives the Hasse diagram of $P$.
\item[Arbitrary (binary) Relation (AR)]  An arbitrary directed graph (or arbitrary binary relation) is given, whose reflexive, transitive closure is then a pre-order whose induced partial order is $P$.  (Equivalently, $P$ is the set of strongly connected components, and $\le$ is the reachability relation between these components.)
\end{description}
The first two (PO and HD) must involve promises that the input satisfies the corresponding constraint, so problems in these categories are posed as promise problems.  Notice that the PO promise is stronger than the HD promise, which is stronger than the AR (vacuous) promise.  So in either the $\EX$ or $\SU$ cases, the complexity of the corresponding problems increases monotonically as $\PO\rightarrow\HD\rightarrow\AR$.

We will ignore some additional subtleties: In the explicit case, is the graph (or relation) given by an adjacency matrix or an array of edge lists?  In the succinct case, should we be able to represent a poset whose size is not a power of $2$?  For example, should we insist on including a second circuit that tells us whether a given binary string represents a point in the poset?  These questions can generally be finessed, and they do not affect any of the results.

\subsection{The decision problems}

The two genres and three approaches above can be combined to give six versions of the basic decision problem for arbitrary posets: the three explicit problems $\PG(\EX,\PO)$, $\PG(\EX,\HD)$, and $\PG(\EX,\AR)$; and the three succinct problems $\PG(\SU,\PO)$, $\PG(\SU,\HD)$, and $\PG(\SU,\AR)$.  We will define just a couple of these, the others being defined analogously.

\begin{definition}\rm
$\PG(\SU,\HD)$ is the following promise problem:
\begin{verse}
\textbf{Input:} A Boolean circuit $C$ with one output and two inputs of $n$ bits each, for some $n$.\\
\textbf{Promise:} $G$ is acyclic, where $G$ is the digraph on $\two^n$ whose edge relation is computed by $C$.\\
\textbf{Question:} Letting $P$ be the poset given by the reachability relation on $G$, is $P$ an $\exists$-game?
\end{verse}
\end{definition}

%

\begin{definition}\rm
$\PG(\EX,\AR)$ is the following promise problem:
\begin{verse}
\textbf{Input:} A digraph $G$ on $n$ nodes.\\
\textbf{Promise:} None.\\
\textbf{Question:} Letting $P$ be the poset given by the reachability
  relation on the strongly connected components of $G$, is $P$ an $\exists$-game?
\end{verse}
\end{definition}

We also can denote subcategories of poset games the same way.  For example, $\Nim(\EX,\HD)$ is the same as $\PG(\EX,\HD)$, but with the additional promise that the poset is a parallel union of chains; for any $k>0$, $\Nim_k(\EX,\HD)$ is the same as $\Nim(\EX,\HD)$ but with the additional promise that there are at most $k$ chains; $\Nfree(\SU,\PO)$ is the same as $\PG(\SU,\PO)$ with the additional promise that the poset is N-free.

\subsection{The first results}

Here are the first lower bounds known for poset games, given roughly in chronological order.  The first four involve \Nim; the first two of these consider explicit games, and the next two consider succinct games.  None of these results is currently published, and we will give sketches of their proofs here.

\begin{theorem}[Wagner, 2009]
$\Nim_4(\EX,\HD)$ is $\LOGSPACE$-hard under $\AC^0$ reductions.
\end{theorem}

The proof reduces from the promise problem ORD (order between vertices), which is known to be complete for $\LOGSPACE$ via quantifier-free projections \cite{Etessami:L,JKMT:GI}.

\begin{proof}
The promise problem ORD (order between vertices) is known to be complete for $\LOGSPACE$ via quantifier-free projections \cite{Etessami:L,JKMT:GI}:
\begin{verse}
\textbf{Input:} A directed graph $G$ on $n$ nodes (given by a binary edge relation $E(G)$) and two distinct vertices $x$ and $y$ of $G$.\\
\textbf{Promise:} $G$ is a single directed path with no cycles.\\
\textbf{Question:} Is $y$ reachable from $x$ in $G$?
\end{verse}
We may assume that both $x$ and $y$ have successors along the path in $G$, say, $s$ and $t$, respectively; otherwise, the problem is trivial.  We can translate any instance $\tuple{G,x,y}$ of ORD into an instance $P$ of $\Nim_4(\EX,\HD)$ (i.e., a dag consisting of at most four disjoint simple paths) so that $y$ is reachable from $x$ if and only if $P$ (considered a poset game) is an $\exists$-game.  We do this as follows: $P$ contains two disjoint copies of $G$, say, $G$ and $G'$, where we label vertices of $G$ with unprimed letters and the corresponding duplicate vertices in $G'$ with primed letters.  We make the following additional changes to $P$:
\begin{itemize}
\item
Remove the edges $(x,s)$ and $(y,t)$ from $E(G)$, and remove the edges $(x',s')$ and $(y',t')$ from $E(G')$.
\item
Add \emph{crossing edges} $(y,t')$ and $(y',t)$ to $E(P)$.
\item
Add two directed paths $p_1\rightarrow p_2\rightarrow\cdots\rightarrow p_n$ and $q_1\rightarrow q_2\rightarrow\cdots\rightarrow q_n$ to $P$, both of length $n$.
\item
Add \emph{connecting edges} $(p_n,v)$ and $(x,q_1)$ to $E(P)$, where $v$ is the initial vertex along the path of $G$.
\end{itemize}
Let $w$ be the final vertex of $G$.  The two possible scenarios for $P$ are shown in Figure~\ref{fig:wagner}.
\begin{figure}
\begin{center}
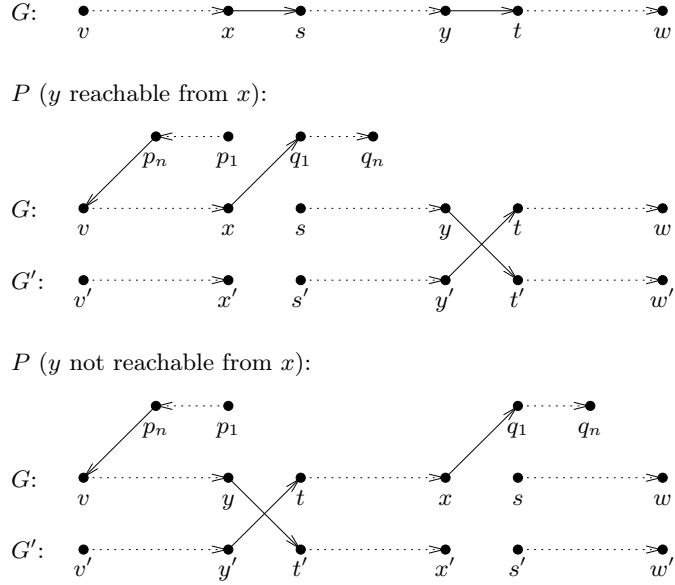
\caption{The construction of $P$ from $G$.  $G$ is shown at the top in the case where $y$ is reachable from $x$.  Shown immediately below is $P$ in this case.  Below that is shown $P$ when $y$ is not reachable from $x$.}\label{fig:wagner}
\end{center}
\end{figure}
If $y$ is reachable from $x$, then we get the \Nim\ game near the top of the figure, whose g-number is of the form $(2n+k)\oplus k$ for some $k$, owing to the two paths on the left (the paths on the right are the same length, so they cancel).  This is nonzero, hence $P$ is an $\exists$-game.  Otherwise, we have the game at the bottom of the figure, and this is clearly a $\forall$-game, consisting of two pairs of paths of equal length.  

The construction of $P$ from $G$ can be done in $\AC^0$, which proves the theorem.
\end{proof}

\begin{theorem}[Thierauf, 2009]
$\Nim_2(\EX,\AR)$ is $\NL$-hard under $AC^0$ reductions.
\end{theorem}

The proof reduces from the reachability problem for directed graphs, which is $\NL$-complete under $\AC^0$-reductions.

\begin{proof}
We reduce from the reachability problem for directed graphs, which is $\NL$-complete under $\AC^0$-reductions:
\begin{verse}
\textbf{Input:} A directed graph $G$ on $n$ nodes (given by a binary edge relation $E(G)$) and two distinct vertices $s$ and $t$ of $G$.\\
\textbf{Question:} Is $t$ reachable from $s$ in $G$?
\end{verse}
Given $G$ as above, we construct a (possibly cyclic) digraph $H$ whose corresponding poset game is an $\exists$-game if and only if $t$ is reachable from $s$ in $G$.  (Recall that a move in a digraph corresponds to removing a vertex and all vertices reachable from it.)  We let $H$ be two disjoint copies of $G$, say, $G$ and $G'$, where $s'$ are $t'$ are the vertices in $G'$ corresponding to $s$ and $t$ in $G$, respectively.  We then add two more edges to $H$: one from $t$ to $s'$ and the other from $t'$ to $s$.  See Figure~\ref{fig:thierauf}.
\begin{figure}
\begin{center}
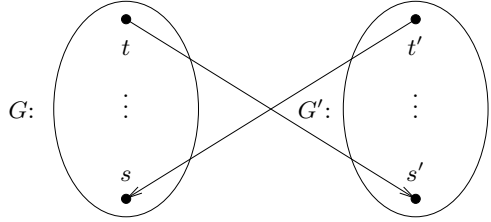
\caption{The graph $H$ constructed from $G$.}\label{fig:thierauf}
\end{center}
\end{figure}
The construction of $H$ from $G$ is clearly $\AC^0$.  If $t$ is reachable from $s$ in $G$, then choosing, say, $s$ removes from $H$ all vertices except those not reachable from either $s$ or $s'$.  This is a winning move, because the remaining graph consists of two disjoint, identical components---one in $G$ and the other in $G'$, and so it is the parallel union of identical subgames, thus a $\forall$-game.

If $t$ is not reachable from $s$ in $G$, then the game $H$ itself consists of two disjoint, identical subgraphs, and so is a $\forall$-game.
\end{proof}

The next result about succinct poset games is straightforward.

\begin{theorem}[F, 2009]
$\Nim_2(\SU,\PO)$ is $\co{\CeqP}$-hard under p-re\-ductions.
\end{theorem}

The idea here is that, for any $L\in\co{\CeqP}$ and any input $x$, we produce two NIM stacks, and $x\in L$ if and only if they are of unequal length.

\begin{proof}
If $L$ is a language in $\co{\CeqP}$, then by standard results in complexity theory (see \cite{FFK:gaps} for example), there exists a positive polynomial $p(n)$ and a polynomial-time predicate $R$ such that, for all $n$ and $x\in\two^n$, we have
\[ x\in L \iff \left|\left\{y\in \two^{p(n)} : R(x,y) \right\}\right| \ne 2^{p(n)-1}\;. \]
Then given $x$ of length $n$, we can construct in polynomial time a Boolean circuit $C_x$ that takes two $p(n)$-bit inputs and produces a one-bit output such that
\[ C_x(y,z) = 1 \iff y \le z \myand R(x,y) = R(x,z) \]
for all $y,z\in\two^{p(n)}$.  The circuit $C_x$ computes a partial order relation on $\two^{p(n)}$ which is the parallel union of two chains.  The size of one chain is the number of $y\in\two^{p(n)}$ such that $R(x,y)$ holds, and the sum of the two sizes is $2^n$.  Thus $x\in L$ if and only if the chains are of unequal size, if and only if the resulting two-stack \Nim\ game is an $\exists$-game.
\end{proof}

\begin{theorem}[F, 2009]\label{thm:pspace-nim}
$\Nim_6(\SU,\HD)$ is $\PSPACE$-hard under p-re\-ductions.
\end{theorem}

The proof uses a result of Cai \& Furst \cite{CF:PSPACE} based on techniques of David Barrington on bounded-width branching programs.  Recall that $S_5$ is the group of permutations of the set $\{1,2,3,4,5\}$.  Their result is essentially as follows:

\begin{theorem}[Cai \& Furst]
For any $\PSPACE$ language $L$, there exists a polynomial $p$ and a polynomial-time computable (actually, log-space computable) function $\sigma$ such that, for all strings $x$ of length $n$ and positive integers $c$ (given in binary), \ $\sigma(x,c)$ is an element of $S_5$, and $x\in L$ if and only if the composition $\sigma(x,1)\sigma(x,2)\sigma(x,2)\cdots\sigma(x,2^{p(n)})$, applied left to right, fixes the element $1$.
\end{theorem}

The idea is that we connect the first five NIM stacks level-by-level via permutations in $S_5$, as well as adding a couple of widgets.  If the product of all the permutions fixes $1$, then we get five NIM stacks of equal length $N$ and one NIM stack of length $N+2$, which is an $\exists$-game by the Sprague-Grundy theorem.  If $1$ is not fixed, then we get four stacks of length $N$ and two of length $N+1$---a $\forall$-game by the same theorem.

\begin{proof}[Proof of Theorem~\ref{thm:pspace-nim}]
Fix $L\in\PSPACE$, and let $p$ and $\sigma$ be as in Cai \& Furst's result above.  For any $x$ of length $n$, we define a directed acyclic graph $G_x$ as follows: $G_x$ has $6\cdot 2^{p(n)}+2$ vertices that come in three types (letting $N = 2^{p(n)}$):
\begin{enumerate}
\item
For $c = 0,1,2,\ldots,N$ and all $k\in\{1,2,3,4,5\}$, \ $u^k_c$ is a vertex of $G_x$.
\item
For $c = 0,1,2,\ldots,N$, \ $v_c$ is a vertex of $G_x$.
\item
$G_x$ has two additional vertices $s$ and $t$.
\end{enumerate}
For convenience, let $\sigma_c$ denote $\sigma(x,c)$.  The graph $G_x$ has three kinds of edges (and no others):
\begin{enumerate}
\item
For $c = 1,2,3,\ldots,N$ and all $k\in\{1,2,3,4,5\}$, \ $(u^k_{c-1},u^{\sigma_c(k)}_c)$ is an edge of $G_x$.
\item
For $c = 1,2,3,\ldots,N$, \ $(v_{c-1},v_c)$ is an edge of $G_x$.
\item
$(s,u^1_0)$ and $(u^1_N,t)$ are edges of $G_x$.
\end{enumerate}
A typical $G_x$ is shown in Figure~\ref{fig:nim-game}.
\begin{figure}
\begin{center}
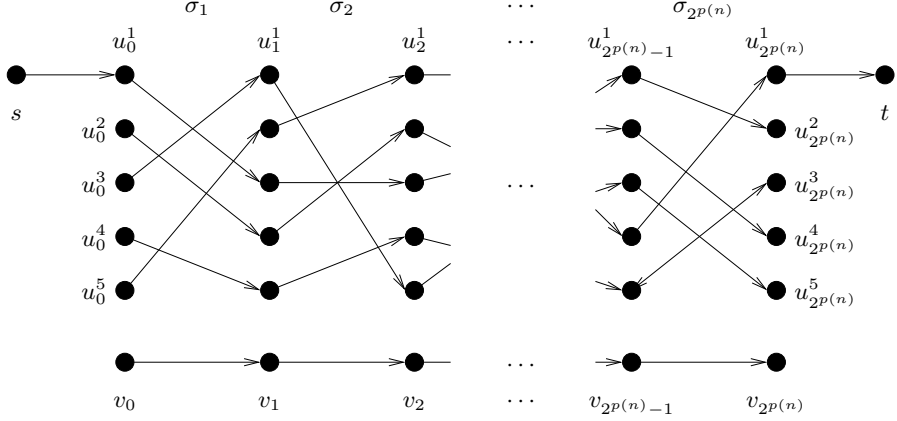
\caption{The graph $G_x$ constructed from $x$.}\label{fig:nim-game}
\end{center}
\end{figure}
The columns of vertices (besides $s$ and $t$) are indexed by $c$ running from $0$ to $N$.  The five rows of $u$-vertices are indexed by $k\in\{1,2,3,4,5\}$.  The $k$'th $u$-vertex in column $c-1$ has one outgoing edge to the $\sigma_c(k)$'th $u$-vertex in column $c$.  Then it is evident that the game $G_x$ consists of six NIM stacks---the first five involving $u$-vertices and the last consisting of the $v$-vertices.  Let $\sigma\in S_5$ be the left-to-right composition $\sigma_1\sigma_2\cdots\sigma_N$.  If $\sigma$ fixes $1$, then $s$ and $t$ lie in the same stack, which thus has length $N+3$, and the other five stacks have length $N+1$.  Otherwise, $s$ and $t$ lie in different stacks, and thus $G_x$ has two stacks of length $N+2$ and four stacks of length $N+1$.  In the former case, $G_x$ is an $\exists$-game and in the latter case, $G_x$ is a $\forall$-game.  This shows that $x\in L$ if and only if $G_x$ is an $\exists$-game.

Since each permutation $\sigma_c$ is computed uniformly in polynomial time, one can easily (time polynomial in $n$) construct a Boolean circuit computing the edge relation on $G_x$ as well as a membership test for $V(G_x)$.  Thus we have a p-reduction from $L$ to $\Nim_6(\SU,\HD)$.
\end{proof}

Although the above results all mention \Nim, the representations we use of a \Nim\ game as a poset are not the natural one.  Therefore, it is better to consider these as lower bounds on N-free poset games, which \emph{are} naturally represented as posets.

The next results regard $N$-free games.  They depend on Adam Kalinich's game outcome-flipping trick.  The trick turns a poset game $A$ into another poset game $\neg A$ with opposite outcome, starting with $A$ and applying series and parallel union operations in a straightforward way.  Here we describe a simplification of the trick due to Daniel Grier:

\noindent Given a poset $A$,
\begin{enumerate}
\item
Let $k$ be any (convenient) natural number such that $2^k \ge |A|$ (that is, $A$ has at most $2^k$ elements).
\item
Let $B \eqdf A/C_{2^k-1}$.
\item
Let $C \eqdf B + C_{2^k}$.
\item
Let $D \eqdf C/C_1$.
\item
Finally, define $\neg A \eqdf D + A$.
\end{enumerate}

Let's check the following

\begin{claim}\label{claim:not-A}
If $\gnum(A) \ne 0$, then $\gnum(\neg A) = 0$.  If $\gnum(A) = 0$, then $\gnum(\neg A) = 2^{k+1}$.
\end{claim}

\begin{proof}
Recall that $\gnum(P) \le |P|$ for any poset $P$, and thus $\gnum(A) \le 2^k$.  By Fact~\ref{fact:series-gnum}, $\gnum(B) = \gnum(A) + 2^k - 1$, so if $\gnum(A) = 0$, then $\gnum(B) < 2^k$, and otherwise, $2^k \le \gnum(B) < 2^{k+1}$, which implies the $(k+1)$st least significant bit position of $\gnum(B)$ is $1$.  By Theorem~\ref{thm:sg}, $\gnum(C) = \gnum(B) \oplus \gnum(C_{2^k}) = \gnum(B) \oplus 2^k$, which is just $\gnum(B)$ with its $(k+1)$st bit flipped.  So if $\gnum(A) = 0$, then clearly, $\gnum(C) = \gnum(B) + 2^k = \gnum(A) + 2^{k+1} - 1 = 2^{k+1} - 1$, and otherwise, $\gnum(C) = \gnum(B) - 2^k = \gnum(A) - 1$.  Next, we have $\gnum(D) = \gnum(C) + 1$, and so $\gnum(D) = 2^{k+1}$ if $\gnum(A) = 0$, and $\gnum(D) = \gnum(A)$ otherwise.  Finally, this gives
\[ \gnum(\neg A) = \gnum(D) \oplus \gnum(A) = \left\{ \begin{array}{ll}
2^{k+1} & \mbox{if $\gnum(A) = 0$,} \\
0 & \mbox{if $\gnum(A) \ne 0$,}
\end{array} \right. \]
and we are done.
\end{proof}

Observe that the size of $\neg A$ is linearly bounded in $|A|$.  In fact, $|\neg A| \le 6|A|$ if $A\ne\emptyset$.

\begin{theorem}[Kalinich~\cite{Kalinich:flip}]\label{thm:Kalinich}
$\Nfree(\EX,\PO)$ is $\NC^1$-hard under $\AC^0$ reductions.
\end{theorem}

\begin{proof}[Proof sketch]
We reduce from the Circuit Value problem for $\NC^1$ circuits with a single output.  Given an $\NC$ circuit $C$ with a single output and whose inputs are constant Boolean values, we produce a poset game $P$ so that $P$ is an $\exists$-game if and only if $C=1$.  We can assume WLOG that all gates in $C$ are either (binary) OR-gates or NOT-gates.  Starting with the input nodes, we associate a poset $P_n$ with every node $n$ in $C$ from bottom up so that the outcome of $P_n$ matches the Boolean value at node $n$.  $P$ is then the poset associated with the output node of $C$.  The association is as follows:
\begin{itemize}
\item
If $n$ is an input node, we set $P_n \eqdf \emptyset$ if $n=0$; otherwise, if $n=1$, we set $P_n \eqdf C_1$.
\item
If $n$ is an OR-gate taking nodes $\ell$ and $r$ as inputs, then we set $P_n \eqdf P_\ell/P_r$.  (Recall Exercise~\ref{ex:OR-gate}.)
\item
If $n$ is a NOT-gate taking node $c$ as input, we set $P_n \eqdf \neg P_c$.
\end{itemize}
This transformation from $C$ to $P$ can be done in (uniform) $\AC^0$, producing a poset of polynomial size, provided $C$ has $O(\log n)$ depth.
\end{proof}

The next theorem is not published elsewhere.

\begin{theorem}[F, 2011]\label{thm:succinct-po}
$\Nfree(\SU,\PO)$ is $\PP$-hard under p-re\-ductions.
\end{theorem}

To prove this theorem, we first need to generalize the Kalinich/Grier construction a bit.

\begin{definition}\label{def:threshold}\rm
For any poset $A$ and any integer $t>0$, define
\[ \threshold{A}{t} \eqdf \frac{(A/C_{2^k - t}) + C_{2^k}}{C_t} + A\;, \]
where $k$ is any convenient natural number (the least, say) such that $2^k > \max(|A|-t,t-1)$.
\end{definition}

Note that $\neg A = \threshold{A}{1}$.  A proof virtually identical to that of Claim~\ref{claim:not-A} shows that
\begin{equation}\label{eqn:threshold}
\gnum(\threshold{A}{t}) = \left\{ \begin{array}{ll}
2^{k+1} & \mbox{if $\gnum(A) < t$,} \\
0 & \mbox{if $\gnum(A) \ge t$.}
\end{array} \right.
\end{equation}

We then use the $\threshold{\cdot}{\cdot}$ operator to polynomially reduce any $\PP$ language to $\Nfree(\SU,\PO)$.  The next fact is routine and needed for the proof of Theorem~\ref{thm:succinct-po}.

\begin{fact}
Given as input a value of $t$ and the succinct representation of a poset $A$, one can build a succinct representation of $\threshold{A}{t}$ in polynomial time.
\end{fact}

\begin{proof}[Proof of Theorem~\ref{thm:succinct-po}]
By standard results in complexity, for any $L\in\PP$, there is a polynomial $p$ and a polynomial-time function $x\mapsto B_x$ mapping inputs to Boolean circuits such that, for all $x$, (i) $B_x$ has $q \eqdf p(|x|)$ many input nodes; and (ii) $x\in L$ if and only if $B_x(y) = 1$ for at least $2^{q-1}$ many inputs $y$.  We can assume WLOG that $q\ge 2$.  Given $B_x$, we can in polynomial time construct a circuit $D_x$ with two input registers of $q$ bits each, such that for all $y,z\in\two^q$, \ $D_x(y,z) = 1$ if and only if either: (a) $y=z$, or (b) $y < z$ and $B_x(y) = B_x(z) = 1$.  Suppose $|\{y : B_x(y)=1\}| = k$.  Then $D_x$ is the succinct $\PO$ representation of the poset $P \eqdf C_k + A_{2^q - k}$, consisting of the parallel union of a chain of length $k$ with an antichain of length $2^q - k$.  Using Theorem~\ref{thm:sg}, we get that $\gnum(P) = \gnum(C_k) \oplus \gnum(A_{2^q - k}) = k \oplus (k \bmod 2)$, the latter quantity being either $k$ or $k-1$, whichever is even.  Now let $T \eqdf \neg\threshold{P}{2^{q-1}}$.  Then $T$ is an $\exists$-game if and only if $\gnum(\threshold{P}{2^{q-1}}) = 0$, if and only if $\gnum(P) \ge 2^{q-1}$, if and only if $k \ge 2^{q-1}$ (note that $2^{q-1}$ is even, because $g \ge 2$), if and only if $x\in L$.  Since $T$ is clearly $N$-free, and a circuit for $T$ can be constructed from $x$ in polynomial time, this shows that $L \le_m^p \Nfree(\SU,\PO)$.
\end{proof}

\subsection{A note on the complexity of the g-number}

Of course, computing the g-number of an impartial game is at least as hard as computing its outcome, the latter just being a test of whether the g-number is zero.  Is the reverse true, i.e., can we polynomial-time reduce computing the g-number to computing the outcome?  For explicitly represented poset games, this is certainly true.  Given an oracle $S$ returning the outcome of any poset game, we get the g-number of a given poset game $G$ as follows: query $S$ with the games $G, G+C_1, G+C_2, \ldots, G + C_n$, where $n$ is the number of options of $G$ (recall that that $C_i$ is a NIM stack of size $i$).  By the Sprague-Grundy theorem (Theorem~\ref{thm:sg}), all of these are $\exists$-games except $G+C_{\gnum(G)}$, which is a $\forall$-game.

What about succinctly represented games?  The approach above can't work, at least for poset games, because the poset has exponential size.  Surprisingly, we can still reduce the g-number to the outcome for succinct poset games in polynomial time, using the threshold construction of Definition~\ref{def:threshold} combined with binary search.  Given a succinctly represented poset $P$ of size $\le 2^n$, first query $S$ with $\threshold{P}{2^{n-1}}$.  If $S$ says that this is an $\exists$-game, then we have $\gnum(P) < 2^{n-1}$; otherwise, $\gnum(P) \ge 2^{n-1}$.  Next, query $S$ with $\threshold{P}{2^{n-2}}$ in the former case and $\threshold{P}{3\cdot 2^{n-2}}$ in the latter case, and so on.  Note that in this reduction, the queries are adaptive, whereas they are nonadaptive for explicitly represented games.

\subsection{$\PSPACE$-completeness}

In this section we sketch the proofs of two recent $\PSPACE$-completeness results for poset game.  The first, by Daniel Grier, is that the outcome problem for general explicit (impartial) poset games is $\PSPACE$-complete \cite{Grier:poset-games}.  The second is a similar result about the complexity of black-white poset games \cite{FGMST:black-white}.

\begin{theorem}[Grier~\cite{Grier:poset-games}]\label{thm:poset-pspace-complete}
Deciding the outcome of an arbitrary finite poset game is $\PSPACE$-complete.
\end{theorem}

\begin{proof}
Membership in $\PSPACE$ is clear.  For $\PSPACE$-hardness, we reduce from $\NK$.  Let $G=(V,E)$ (a simple undirected graph) be an arbitrary instance of $\NK$.  By altering the graph slightly if necessary without changing the outcome of the game, we can assume that $|E|$ is odd and that for every $v\in V$ there exists $e\in E$ not incident with $v$.  We can do this by adding two disjoint cliques to $G$---either two $K_2$'s or a $K_2$ and a $K_4$, whichever of these options results in an odd number of edges.  We then construct the following three-level poset $P$ from $G$:
\begin{itemize}
\item
The points of $P$ are grouped into three disjoint antichains, $A$, $B$, and $C$, with $A$ being the set of minimal points, $C$ the maximal points, and $B$ the points intermediate between $A$ and $C$.
\item
For each edge $e\in E$ there correspond unique points $c_e \in C$ and $a_e\in A$, and vice versa.
\item
We let $B \eqdf V$.
\item
For each edge $e = \{v_1,v_2\}$ and $b\in B$, we have $b < c_e$ iff $b = v_1$ or $b=v_2$, and $a_e < b$ iff this is not the case, i.e., iff $b \ne v_1$ and $b \ne v_2$.  This is illustrated in Figure~\ref{fig:grier}.
\end{itemize}
\begin{figure}
\begin{center}
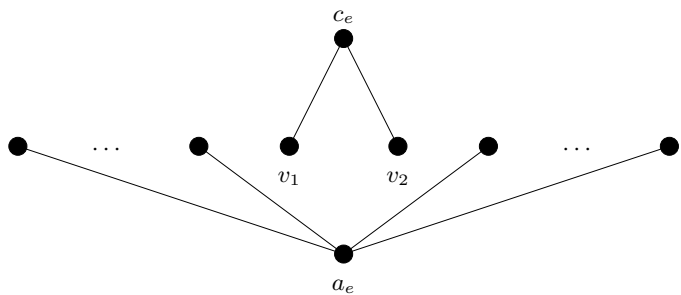
\caption{The $<$ relations in $P$ obtained from the edge $e = \{v_1,v_2\}$ in $G$.}\label{fig:grier}
\end{center}
\end{figure}
This construction can clearly be done in polynomial time, given $G$.

Now we show the outcomes are the same for the two games:  The winning player in the game $G$---Left, say, who may play first or second---can also win in the game $P$ by playing the $B$-points corresponding to the vertices she plays to win in $G$, for as long as Right does the same.  When Right first deviates from this type of play (and he must, because he loses the game $G$), Left can respond as follows:
\begin{itemize}
\item
If Right plays some $v\in B$ adjacent (in $G$) to some other $u\in B$ already played, then Left plays $a_{\{u,v\}}$, resulting in an empty poset.
\item
If Right plays $c_e\in A$ for some $e \in E$, then Left plays $a_e$, leaving an antichain of size $2$.
\item
If Right plays $a_e\in A$ for some $e = \{u,v\}\in E$, then
\begin{itemize}
\item
if either $u$ or $v$ has already been played, then Left plays the other vertex, leaving only an even number of points in $P$, all of them in $A$, and
\item
if neither $u$ nor $v$ has been played, then Left plays $c_e$, leaving $u,v\in B$ and an even number of points in $A$.
\end{itemize}
In the latter case, if Right then plays either $u$ or $v$, then Left plays the other vertex.  Otherwise, if Right plays some $a_{e'}$, then this removes at least one of $u$ and $v$, say, $u$.  Then Left plays some $a_{e''}$ where $e''$ is not incident to $v$, thus removing $v$ (if it still remains) and leaving an even number of points in $P$, all of them in $A$.
\end{itemize}
Thus the winner of $G$ is the same as the winner of $P$.
\end{proof}

Finally, we turn to the complexity of black-white poset games.  The next theorem is the first $\PSPACE$-hardness result for a numeric game.

\begin{theorem}\label{thm:bw-pspace-complete}
Determining the outcome of a black-white poset game is $\PSPACE$-complete.
\end{theorem}

\begin{proof}[Proof sketch]
Membership in $\PSPACE$ is straightforward.  For hardness, we reduce from TQBF\@.  We present the reduction in detail and briefly describe optimal strategies for the winning players, but we do not show correctness.  See \cite{FGMST:black-white} for a full proof.

Suppose we are given a fully-quantified boolean formula $\p$ of the form $\exists x_1 \forall x_2 \exists x_3 \cdots \exists x_{2n-1} \forall x_{2n} \exists x_{2n+1} f(x_1, x_2, \ldots, x_{2n+1})$, where $f = c_1 \wedge c_2 \wedge \cdots \wedge c_m$ is in cnf with clauses $c_1, \ldots, c_m$.  We define a two-level black-white poset (game) $X$ based on $\p$ as follows:
\begin{itemize}
\item $X$ is divided into sections. There is a section (called a \emph{stack}) for each variable, a section for the clauses (the \emph{clause section}), and a section for fine-tuning the balance of the game (\emph{balance section}). 
\item The $i$th stack consists of a set of incomparable \emph{waiting nodes} $W_i$ above (i.e., greater than) a set of incomparable \emph{choice nodes} $C_i$.  We also have a pair of \emph{anti-cheat nodes}, $\alpha_i$ and $\beta_i$, on all stacks except the last stack. For odd $i$, the choice nodes are white, the waiting nodes are black, and the anti-cheat nodes are black. The colors are reversed for even $i$. 
\item The set of choice nodes $C_i$, consists of eight nodes corresponding to all configurations of three bits (i.e., $000, 001, \ldots, 111$), which we call the \emph{left bit}, \emph{assignment bit} and \emph{right bit} respectively.
\item The number of waiting nodes is $\card{W_i} = (2n+2-i)M$, where $M$ is the number of non-waiting nodes in the entire game.  It is important that $\card{W_i} \geq \card{W_{i+1}} + M$. 
\item The anti-cheat node $\alpha_i$ is above nodes in $C_i$ with right bit $0$ and nodes in $C_{i+1}$ with left bit $0$. Similarly, $\beta_i$ is above nodes in $C_i$ with right bit $1$ and nodes in $C_{i+1}$ with left bit $1$. 
\item The \emph{clause section} contains a black \emph{clause node} $b_j$ for each clause $c_j$, in addition to a black \emph{dummy node}. The clause nodes and dummy node are all above a single white \emph{interrupt node}.  The clause node $b_j$ is above a choice node $z$ in $C_i$ if the assignment bit of $z$ is $1$ and $x_i$ appears positively in $c_j$, or if the assignment bit of $z$ is $0$ and $x_i$ appears negatively in $c_j$. 
\item The balance section or \emph{balance game} is incomparable with the rest of the nodes. The game consists of eight black nodes below a white node, which is designed to have numerical value $-7\tfrac{1}{2}$.  All nodes in this section are called \emph{balance nodes}. 
\end{itemize}
The number of nodes is polynomial in $m$ and $n$, so the poset can be efficiently constructed from $\p$.

\begin{figure}
\begin{center}
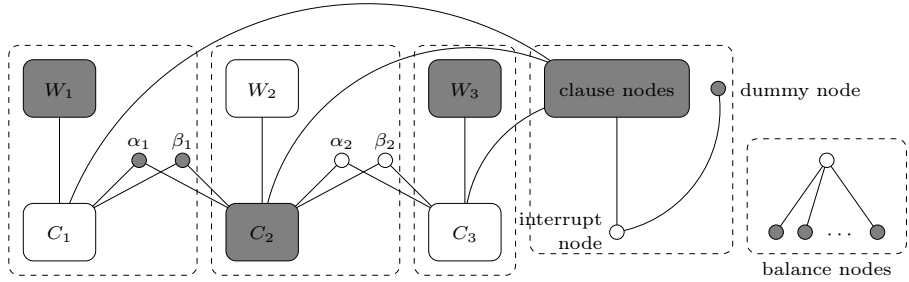
\caption{An example game with three variables ($n = 1$). Circles represent individual nodes, blobs represent sets of nodes, and $\chi$ is the set of clause nodes.  An edge indicates that some node in the lower level is less than some node in the upper level.   The dotted lines divide the nodes into sections (stacks, clause section and balance section).}\label{fig:bw-game}
\end{center}
\end{figure}
A sample construction is shown in Figure~\ref{fig:bw-game}.  The idea is that players take turns playing choice nodes, starting with White, and the assignment bits of the nodes they play constitute an assignment of the variables, $x_1, \ldots, x_{2n+1}$.  The assignment destroys satisfied clause nodes, and it turns out that Black can win if there remains at least one clause node.  The waiting nodes and anti-cheat nodes exist to ensure players take nodes in the correct order.  The interrupt node and dummy node control how much of an advantage a clause node is worth (after the initial assignment), and the balance node ensures the clause node advantage can decide whether White or Black wins the game.  One can show that White (i.e., Right) can force a win when playing first if and only if the formula is true. 

Suppose that White and Black agree to play choice nodes in order, thus producing a truth assignment $a_1,a_2,\ldots$ via the assignment bits.  The other bits are arbitrary, but players would do well to choose each left bit to preserve the remaining anti-cheat node in the previous stack, starting with the second move (so Black preserves a black anti-cheat node in stack~$1$, White an anti-cheat node in stack~$2$, etc.).  This continues until White plays a choice node in $C_{2n+1}$.  At this point, all the variables have been assigned, but there are still points in $X$; we assume the players continue under optimal play. 

Assuming both players stick to the agreement, one can show that White wins (under optimal play) if and only if $\p$ is true.  The rest of the proof in \cite{FGMST:black-white} shows that either player can win if the other player violates the agreement (``cheats'').  Here, we only describe here what to do when your opponent cheats.

We think of the game as having two phases.  The first phase ends when the players have taken at least one node from each $C_i$.  The second phase begins when the first phase ends, and lasts until the end of the game. If the players stick to the agreement as described above, then the last move in the first phase coincides with White setting the truth value $a_{2n+1}$ by playing in $C_{2n+1}$.

\subsubsection{Phase one strategy}

In phase one, our strategy for White is the same as our strategy for Black: play fair (no cheating!) until our opponent cheats.  If our opponent cheats then reply according to the following rules, and continue to reply according to these rules for future moves.  For the following rules, stack~$i$ is the leftmost stack containing waiting nodes of our color (i.e., we are waiting for our opponent to play in stack~$i$).
\begin{itemize}
\item If the opponent moves in $C_j$, then
\begin{itemize}
\item if $j = 2n+1$, then take a waiting node in $W_i$, else
\item if it is their first move in $C_j$, reply in $C_{j+1}$.  Choose a node that saves one of your anti-cheat nodes and destroys your opponent's anti-cheat nodes where possible.  The assignment bit of your reply will not matter.
\item if it is not their first move in $C_j$, take a waiting node in $W_i$. 
\end{itemize}
\item If the opponent takes a waiting node in $W_{j+1}$ then take a node in $W_j$. 
\item If the opponent takes an anti-cheat node, a clause node, the dummy node, the interrupt node, or a balance node then take a waiting node in $W_i$. 
\end{itemize}
Observe that we take a waiting node in $W_j$ if the opponent takes a non-waiting node (this can happen at most $M$ times) or takes a waiting node in $W_{j+1}$.  By construction, $\card{W_j} \geq M + \card{W_{j+1}}$, so we cannot run out of waiting nodes.  Similarly, we only take a node in $C_{j+1}$ when the opponent takes their first node from $C_j$, so we have all eight nodes to choose from when we play in $C_{j+1}$.  In other words, the strategy never asks us to take a node that isn't there; the reply moves are always feasible.

\subsubsection{Phase two strategy}

Let $H$ be the black-white poset game at the start of phase two, and let $k$ be the number of surviving clause nodes in $H$.  Assuming no cheating in phase one, each player took exactly one choice node from each stack in phase one, and since there are more white $C_i$'s, Black has the first move in phase two.  The waiting nodes in $W_i$ are gone because some node in $C_i$ is missing for all $i$.  Similarly, there is at most one anti-cheat node in each stack, since at least one was destroyed by the missing choice nodes on either side.  

Our description of phase two consists of a series of facts:
\begin{itemize}
\item A player can always avoid destroying their own anti-cheat nodes in $H$, and therefore we may assume it is impossible for a player to destroy their own anti-cheat node. This gives us a new, equivalent game $H' \approx H$, where in $H'$ the anti-cheat node in stack~$i$ is incomparable with all the choice nodes in stack~$i+1$, for $i=1,\ldots,2n$. 
\item It is optimal (in $H'$) for White to take the interrupt node after Black's first move, as long as the dummy node is intact. 
\item It is optimal for Black to take a clause node on his first move in $H'$, if one exists. 
\end{itemize}
It follows that the clause nodes are gone by Black's second move in $H'$.  Let $J$ be $H'$ with its clause section removed.  Then every section (i.e., each stack and the balance section) in $J$ is incomparable with the rest of $J$.  This means we can write $J$ as the sum of much simpler games:
\[ J = J_1 + J_2 + \cdots + J_{2n} + J_{2n+1} + B\;, \]
where $J_i$ is the $i$'th stack component of $J$ and $B$ is the balance nodes.

$J_i$ has numerical value $\pm 7$ without an anti-cheat node, and $\pm 6\tfrac{1}{2}$ with an anti-cheat node, where the sign is $(-1)^i$.  Note that the last stack, $i = 2n+1$, does not contain an anti-cheat node, and so its value is $-7$.  The balance section $B$ has value $7\frac{1}{2}$ by construction (see Exercise~\ref{ex:sample-bw-games}), so if all the anti-cheat nodes survive,
\[ v(J) = \sum_{i=1}^{2n+1} v(J_i) + v(B) = 6{\textstyle\frac{1}{2}} \sum_{i=1}^{2n} (-1)^i - 7 + 7{\textstyle\frac{1}{2}} = \frac{1}{2}\;. \]
We call this the \emph{baseline} value.

If $\p$ is true (and Black does not cheat), then White manages to clear away all the clause nodes in phase one.  So then $H' = J + C$, where $C$ is just the interrupt node and dummy node.  Since $v(C) = -\frac{1}{2}$, we get $v(H') = 0$, which is a win for White (because Black plays first in $H'$).  If Black cheats, one can show that she does so at the cost of one of her anti-cheat nodes, which again reduces $v(H')$ to $0$, a win for White.

If $\p$ is false (and White does not cheat), then White cannot clear all the clause nodes in phase one.  Black then plays a clause node to start phase two, after which White plays the interrupt node.  The remaining game is $J$, with no clause section and all anti-cheat nodes, whose value is $\frac{1}{2}$, a win for Black.  If White tries to cheat, then he may be able to destroy all clause nodes, but at the expense of at least one white anti-cheat node.  The clause section subtracts $\frac{1}{2}$, but losing an anti-cheat node adds $\frac{1}{2}$, bringing us back to the baseline $\frac{1}{2}$, a win for Black.
\end{proof}
 
\section{Open questions}
\label{sec:open}

Are there interesting games whose complexity is complete for a subclass of $\PSPACE$?  The natural black-white version of \GenCol\ is complete for the class $\PTIME^{\NP[\log]}$ (that is, the class of decision problems computable in polynomial time with $O(\log n)$ many oracle queries to an $\NP$ language), but the game itself and the reasons for its complexity are not so interesting.  In this version, each uncolored node is reserved (``tinted'') for being colored one or the other color, e.g., some node $u$ can only be colored black, while some other node $v$ can only be colored white, and so on for all the nodes.  Then the outcome of this game depends only on which subgraph (the black-tinted nodes or the white-tinted nodes) contains a bigger independent set.  Given two graphs $G_1$ and $G_2$, the problem of determining whether $G_1$ has a bigger independent set than $G_2$ is known to be complete for $\PTIME^{\NP[\log]}$ \cite{ST:independent-sets}.

Fix a natural number $k>2$.  For poset games of bounded width $k$, defined in Section~\ref{subsubsec:bwpg}, is there an algorithm running in time $o(n^k)$?

Grier's proof that the poset game decision problem is $\PSPACE$-complete (Theorem~\ref{thm:poset-pspace-complete}) constructs posets having three levels, that is, whose maximum chain length is three.  What about two-level games in general?  Those having a single maximum or a single minimum element are easily solved.  What is the complexity of those with more than one minimum and more than one maximum?  Certain subfamilies of two-level posets have g-numbers that show regular patterns and are easily computed, or example, games where each element is above or below at most two elements, as well as parity-uniform games (see Definition~\ref{def:parity-uniform} and Theorem~\ref{thm:parity-uniform}) \cite{FGKT:two-level}.  Despite this, we conjecture that the class of all two-level poset games is $\PSPACE$-complete, but are nowhere near a proof.  Are there larger subfamilies of the two-level poset games that are in $\PTIME$?

A more open-ended goal is to apply the many results and techniques of combinatorial game theory, as we did in Theorem~\ref{thm:bw-pspace-complete}, to more families of games.

Finally, we mention a long-standing open problem about a specific infinite poset game: What is the outcome of the game $\nums^3 - \{(0,0,0)\}$, where $(x_1,x_2,x_3) \le (y_1,y_2,y_3)$ iff $x_i \le y_i$ for all $i\in\{1,2,3\}$?

\bibliography{/Users/steve/Dropbox/research/bib/master}

\end{document}